\documentclass[final,journal, twoside, 10pt]{IEEEtran}
\usepackage[hyperref]{Bamble}
\usepackage{Bmacros}

\begin{document}
\title{Fast Polarization for Processes with Memory} 
\date{}
\author{\IEEEauthorblockN{Boaz~Shuval, Ido~Tal\\
Department of Electrical Engineering,\\
Technion, Haifa 32000, Israel.\\
Email: \{\texttt{bshuval@campus}, \texttt{idotal@ee}\}\texttt{.technion.ac.il}}
\thanks{An abbreviated version of this article will be submitted to ISIT 2018.}}

\maketitle

\begin{abstract}
	Fast polarization is crucial for the performance guarantees of polar codes. In the memoryless
    setting, the rate of polarization is known to be exponential in the square root of the block
    length. A complete characterization of the rate of polarization for models with memory has been
    missing. Namely, previous works have not addressed fast polarization of the high entropy set under memory.
    We consider polar codes for processes with memory that are characterized by an underlying ergodic finite-state Markov chain. We show that the rate of polarization for these processes is the same as in the memoryless setting, 
	both for the high and for the low entropy sets. 
\end{abstract}
\iftoggle{IEEEtran}{
}{}

\section{Introduction}
\iftoggle{IEEEtran}{\IEEEPARstart{M}{emory}}{Memory} 
is prevalent in many communication scenarios. Examples include finite-state channels (FSCs) such as intersymbol interference channels and correlated fading channels, and coding for input-constrained systems. 
In this research we show that polar codes can be used directly 
for a large class of scenarios with memory. This allows one to leverage the attractive properties of
polar codes --- such as low complexity encoding and decoding, explicit construction, and sound
theoretical basis --- for scenarios with memory. 

A fundamental problem of information theory is estimating a block $\rv{X}_1^N = (\rv{X}_1, \rv{X}_2,
\ldots, \rv{X}_N)$ from observations $\rv{Y}_1^N = (\rv{Y}_1,\rv{Y}_2,\ldots, \rv{Y}_N$). In a channel-coding scenario, $\rv{X}_1^N$ may be the input to a channel and $\rv{Y}_1^N$ its output. In a source-coding scenario, $\rv{X}_1^N$ may be an information source to be compressed and $\rv{Y}_1^N$ observations available to the decompressor. In either case, there is redundancy in $\rv{X}_1^N$: added redundancy in channel coding, or removed redundancy in source coding. A good channel code needs to add the least amount of redundancy while still allowing for correct decoding, whereas a good source code eliminates as much redundancy as possible while still allowing reconstruction subject to a distortion criterion. 

Polar codes~\cite{Arikan_2009} were first developed for binary-input, symmetric, memoryless,
channels. They provide a systematic framework to handle this fundamental problem. They are block
codes, whose encoding operation consists of an explicit invertible transformation between
$\rv{X}_1^N$ and $\rv{U}_1^N$. A portion of $\rv{U}_1^N$ is revealed to the decoder or decompressor.
The decoder employs \emph{successive cancellation} (SC) decoding, recovering $\rv{U}_1^N$
incrementally: first $\rv{U}_1$, then $\rv{U}_2$, and so on. Each successive decoding operation uses
the observations $\rv{Y}_1^N$ and the outcome of the previous decoding operations as well as the
revealed portion of $\rv{U}_1^N$. The polarization phenomenon implies that for large enough $N$, the decoding operations polarize to two sets: a `low entropy' set and a `high entropy' set. These sets can be determined beforehand, and prescribe which portion of $\rv{U}_1^N$ to reveal to the decoder or decompressor. 

The \emph{rate} of polarization is particularly important for the analysis of polar codes. Their error-free performance at any achievable rate is due to polarization happening sufficiently fast. Fast polarization to the low entropy set for the memoryless setting was established in~\cite[Theorem 2]{Arikan_2009},~\cite{Arikan_Telatar_2009}.  

Remarkably, polar codes were extended to a plethora of other memoryless scenarios, including non-binary channels~\cite{sasoglu_nonbinary,sasoglu_thesis}, source coding~\cite{Korada_source,arikan_2010_source}, wiretap channels~\cite{Hof_wiretap,Mahdavifar_wiretap}, asymmetric channels and sources~\cite{Honda_Yamamoto_2013}, and more. See the survey paper~\cite[Section IV]{Arikan_challenges} 
for a large list of extensions and applications. 
Many of these applications are contingent upon fast polarization to the high-entropy set; for memoryless settings, this was established in~\cite{Korada_source}. 

The main tools used for polar code analysis in the memoryless case are the focus of \Cref{sec_polar toolbox}. In particular, we present Ar\i{}kan's probabilistic approach, which is at the heart of many polarization results. It is this approach that we extend to settings with memory.

The study of polar codes for scenarios with memory began with~\cite[Chapter 5]{sasoglu_thesis}.
\c{S}a\c{s}o\u{g}lu  was able to show that polarization indeed occurs for a certain class of
processes with memory. In the subsequent work~\cite{sasoglu_2016} (see also the journal
version,~\cite{sasoglu_Tal_mem}), the authors were
able to prove polarization for a more general class of processes with memory.  
One  advancement made in that paper was regarding the rate of polarization under memory. The authors showed that polarization to the low entropy set is fast even for processes with memory.  Fast polarization to the high entropy set was not addressed.

A practical decoding algorithm for polar codes for FSCs was suggested in~\cite{Wang_2015} (see also~\cite{wang2014joint} for an earlier version, specific to intersymbol interference channels).
This algorithm is an extension of SC decoding, taking into account the underlying state structure. Its increase in complexity relative to the complexity of SC decoding is polynomial with the number of states. Thus, it is practical for a moderate number of states. The authors also showed~\cite[Theorem 3]{Wang_2015} that their elegant scheme from~\cite{Honda_Yamamoto_2013} can be applied to models with memory. To this end, they required the additional assumption of fast polarization both to the low and high entropy sets. 

This paper completes the picture. We show that for a large class of processes with memory, polarization is fast both to the low entropy and high entropy sets. Fast polarization to the low entropy set will follow from a specialization of~\cite{sasoglu_2016}. Fast polarization to the high entropy set, \Cref{thm_fast polarization of Z to 1}, is the main result of this paper. 
 Consequently, polar codes can be used in settings with memory  
with vanishing error probability.

Specifically, we consider stationary processes whose memory can be encompassed by an underlying
finite-state ergodic\footnote{I.e., aperiodic and irreducible.} Markov chain. This Markov chain
governs the joint distribution of $\rv{X}_1^N$ and $\rv{Y}_1^N$, and is assumed to be hidden. The
model is described in detail in \Cref{sec_FAIM}. This family of processes includes, as special
cases, finite-state Markov channels~\cite[Chapter 4.6]{Gallager} with an ergodic state sequence,
discrete ergodic sources with finite memory, and many input-constrained systems (e.g.,
$(d,k)$-runlength limited (RLL) constraint~\cite{marcus_roth_siegel}, with and without noise).

The tools we develop for this family of processes with memory are the subject of \Cref{sec_fast polarization}.
Our tools mirror those used in the memoryless construction. Thus, we expect that this addition to the `polar toolbox' will enable natural adaptation of many polar coding results to settings with memory.

\section{Notation}\label{sec_notation conventions and reminders}
A set of elements is denoted as a list in braces, e.g., $\{1,2,\ldots,L\}$. The number of elements in a set $A$ is denoted by $|A|$. The disjoint union of two sets $A_0,A_1$ is denoted by $A_0 \cupdot A_1$. To use this notation, $A_0$ and $A_1$ must indeed be disjoint.  
Open and closed intervals are denoted by $(a,b)$ and $[a,b]$, respectively. 

We denote $y_j^k = \begin{bmatrix} y_j &y_{j+1}& \cdots & y_k \end{bmatrix}$ for $j<k$. For an arbitrary set of indices $F$ we denote $y_F = \{y_j, j \in F\}$. 

In a summation involving multiple variables, if only one variable is being summed, we will make this explicit by underlining it. For example, in $\sum_{\underline{a}\neq b}f(a,b)$ we sum over the values of $a$ that are different than $b$, and $b$ is fixed. In particular, $\sum_{a\neq b}f(a,b) = \sum_b \sum_{\underline{a}\neq b} f(a,b)$. 

For a sequence of binary numbers $B_1,B_2,\ldots, B_n$ we define $(B_1B_2\cdots B_n)_2 \triangleq
\sum_{j=1}^n B_j 2^{n-j}$. Thus, the rightmost digit $B_n$ is the least significant bit. Addition of binary numbers is assumed to be an XOR operation (i.e., modulo-2 addition).

The probability of an event $A$ is denoted by $\Probi{A}$. Random variables are denoted using a sans-serif font, e.g., $\rv{X}$ and their realizations using lower-case letters, e.g., $x$. The distribution of random variable $\rv{X}$ is denoted by $\pr{X}{} = \pr{X}{x}$. When marginalizing distributions, we will sometimes use the shorthand  $\sum_x \pr{X,Y}{} \equiv \sum_x \pr{X,Y}{x,y}$; the summation variable will denote which random variable is being marginalized. The expectation of $\rv{X}$ is denoted by $\Exp{\rv{X}}$.

\section{The Polar Toolbox}\label{sec_polar toolbox}

\subsection{Various Parameters of Distributions} \label{sec_distribution parameters}
In this section we introduce several parameters that may be computed from the joint distribution of two random variables:  probability of error, Bhattacharyya parameter, conditional entropy, and total variation distance. These parameters are useful for the analysis of polar codes. These parameters are \emph{not} random variables; they are deterministic quantities computed from the joint distribution. 

Consider a pair of random variables $(\rv{U},\rv{Q})$ with  joint distribution $\pr{U,Q}{u,q} = \pr{Q}{q}\pr{U|Q}{u|q}$. The random variable $\rv{U}$ is binary\footnote{This assumption is for the sake of simplicity. See \Cref{rem_non binary extension} at the end of this subsection for a discussion of the implications of non-binary $\rv{U}$.} and $\rv{Q}$ is some observation dependent on $\rv{U}$ that takes values in a finite alphabet $\mathcal{Q}$.

\begin{definition}[Probability of error]
The \emph{probability of error} $\Pe{\rv{U}|\rv{Q}}$ of optimally estimating $\rv{U}$ from the observation $\rv{Q}$, in the sense of minimizing the probability of error, is given by
\begin{equation*}\begin{split} \Pe{\rv{U}|\rv{Q}}  &= \sum_q \min\{\pr{U,Q}{0,q}, \pr{U,Q}{1,q}\} \\
	&= 	\sum_q \pr{Q}{q} \min\{\pr{U|Q}{0|q}, \pr{U|Q}{1|q}\}. 
\end{split} \label{eq_def of Pe} \end{equation*}
\end{definition}

\begin{definition}[Bhattacharyya parameter] The \emph{Bhattacharyya parameter} of $\rv{U}$ given $\rv{Q}$, $\BP{\rv{U}|\rv{Q}}$, is defined as
\begin{equation} \begin{split} \BP{\rv{U}|\rv{Q}} &= 2 \sum_q \sqrt{\pr{U,Q}{0,q}\pr{U,Q}{1,q}} \\
	&= 	 2 \sum_q \pr{Q}{q}\sqrt{\pr{U|Q}{0|q}\pr{U|Q}{1|q}}.
\end{split} \label{eq_def of Z} \end{equation}
\end{definition}

\begin{definition}[Total Variation Distance]
	The \emph{total variation distance} of $\rv{U}$ given $\rv{Q}$, $\TV{\rv{U}|\rv{Q}}$, is defined as 
	\begin{equation} 
		\begin{split}
		\TV{\rv{U}|\rv{Q}} &= \sum_q \left| \pr{U,Q}{0,q} - \pr{U,Q}{1,q} \right|	  \\
							&= \sum_q \pr{Q}{q} \left| \pr{U|Q}{0|q} - \pr{U|Q}{1|q} \right|.  
		\end{split} \label{eq_def of TV}
		\end{equation}
	\end{definition}

    The parameters defined above all required that $\rv{U}$ be binary. They can be extended to the
    non-binary case, as described in Appendix~\ref{app_non binary extension}. A final parameter we will use is the conditional entropy. Unlike the other parameters, the conditional entropy is also defined when $\rv{U}$ takes values in an arbitrary finite alphabet $\mathcal{U}$, not necessarily binary. 
\begin{definition}[Conditional Entropy]
The \emph{conditional entropy} of $\rv{U}$ given $\rv{Q}$, $\ENT{\rv{U}|\rv{Q}}$, is defined as 
\begin{equation}\begin{split} 
\ENT{\rv{U}|\rv{Q}} &= -\sum_{q} \sum_u \pr{U,Q}{u,q}\log_2\frac{\pr{U,Q}{u,q}}{\sum_u\pr{U,Q}{u,q}} \\&= -\sum_{q} \pr{Q}{q} \sum_u \pr{U|Q}{u|q} \log_2 \pr{U|Q}{u|q}.
\end{split} \label{eq_def of condent}	
\end{equation}

\end{definition}

It is easily seen that all four parameters take values in $[0,1]$ when $\rv{U}$ is binary. They are all related, as established in the following lemma. 	
\begin{lemma} \label{lem_TV distance bounds}
	The total variation distance, probability of error, conditional entropy, and Bhattacharyya parameter  are related by
\begin{subequations}\label{eq_TV bounds}
    \begin{align} \TV{\rv{U}|\rv{Q}} &= 1-2\Pe{\rv{U}|\rv{Q}} \geq 1- \ENT{\rv{U}|\rv{Q}}, \label{eq_TV relation 1} \\ 
        \BP{\rv{U}|\rv{Q}}^2 &\leq \ENT{\rv{U}|\rv{Q}} \leq \BP{\rv{U}|\rv{Q}}, \label{eq_bounds on
        BP}  \\
	 \TV{\rv{U}|\rv{Q}} &\leq  \sqrt{1-\BP{\rv{U}|\rv{Q}}^2} \leq \sqrt{1-\ENT{\rv{U}|\rv{Q}}^2}. \label{eq_bounds on TV upper}	
	\end{align}
\end{subequations}
\end{lemma} 
The proof of \Cref{lem_TV distance bounds} is relegated to Appendix~\ref{app_proof of TV lemma}. We
note that the right-most inequality of~\eqref{eq_bounds on BP} was also shown
in~\cite[Proposition 2]{arikan_2010_source} and the left-most inequality
of~\eqref{eq_bounds on TV upper} was also shown in~\cite[Appendix A]{Arikan_2009}; our proof of the
latter is more general. 
Due to~\eqref{eq_TV relation 1}, we shall concentrate in the sequel on $\TV{\rv{U}|\rv{Q}}$ rather than $\Pe{\rv{U}|\rv{Q}}$.

In~\cite{arikan_2010_source}, Ar\i{}kan  used the inequality
\begin{equation}	\BP{\rv{U}|\rv{Q}}^2 \leq \ENT{\rv{U}|\rv{Q}} \leq \log_2 (1+\BP{\rv{U}|\rv{Q}})\label{eq_arikan bounds between ent and TV} \end{equation}
  to show that if the Bhattacharyya parameter approaches $0$ or $1$ then the conditional entropy approaches $0$ or $1$ as well and vice versa.  
  An alternative proof of this can be had by~\eqref{eq_bounds on BP}. 
  This yields 
 \[\BP{\rv{U}|\rv{Q}}^2 \leq \ENT{\rv{U}|\rv{Q}} \leq \BP{\rv{U}|\rv{Q}} \leq \sqrt{\ENT{\rv{U}|\rv{Q}}},\]
which indeed implies that the Bhattacharyya parameter and conditional entropy approach $0$ and $1$ in tandem. This inequality is tighter than~\eqref{eq_arikan bounds between ent and TV}; however, as discussed in Appendix~\ref{app_non binary extension}, an advantage of inequality~\eqref{eq_arikan bounds between ent and TV} is that it has a natural  extension to the case where $\rv{U}$ is non-binary.    

An additional consequence of \Cref{lem_TV distance bounds} is that (a) if $\BP{\rv{U}|\rv{Q}} \to 0$ or $\ENT{\rv{U}|\rv{Q}} \to 0$ then $\TV{\rv{U}|\rv{Q}} \to 1$ and (b) if $\BP{\rv{U}|\rv{Q}} \to 1$ or $\ENT{\rv{U}|\rv{Q}} \to 1$ then $\TV{\rv{U}|\rv{Q}} \to 0$. 

\begin{remark}\label{rem_TV and Pe}
    By combining~\eqref{eq_TV relation 1} and~\eqref{eq_bounds on BP} we obtain 
    
    \[1-2 \Pe{\rv{U}|\rv{Q}} \geq 1- \ENT{\rv{U}|\rv{Q}} \geq 1- \BP{\rv{U}|\rv{Q}}.\]
    Rearranging, we obtain the well-known bound, $\Pe{\rv{U}|\rv{Q}} \leq \BP{\rv{U}|\rv{Q}}/2$. 
\end{remark}

The definitions above naturally extend to the case where instead of $\rv{Q}$ there are multiple random variables related to $\rv{U}$. For example, consider a triplet of random variables $(\rv{U},\rv{Q},\rv{S})$ with joint distribution $\pr{U,Q,S}{u,q,s}$ such that $\rv{U}$ is binary and $\rv{Q},\rv{S}$ take values in finite alphabets $\mathcal{Q},\mathcal{S}$. We call $\rv{S}$ the `state'. Then, 
\[ \TV{\rv{U}|\rv{Q},\rv{S}} = \sum_{q,s} |\pr{U,Q,S}{0,q,s} - \pr{U,Q,S}{1,q,s}|;\] 
 the remaining parameters are similarly extended. We say that $\TV{\rv{U}|\rv{Q},\rv{S}}$ is a \emph{state-informed} (SI) version of $\TV{\rv{U}|\rv{Q}}$. 

How do the SI parameters compare to their non-SI counterparts? 
For the entropy, the answer lies in~\cite[Theorem 2.6.5]{cover_thomas}, the well known property that conditioning reduces entropy. In the following lemma, proved in Appendix~\ref{app_proof of TV lemma}, we consider the other parameters as well.
\begin{lemma}\label{lem_effect of conditioning}
Let $(\rv{U},\rv{Q},\rv{S})$ be a triplet of	 random variables with joint distribution $\pr{U,Q,S}{u,q,s}$. Then 
\begin{subequations}\label{eq_conditioning reduces}
\begin{align}
	\TV{\rv{U}|\rv{Q}} &\leq \TV{\rv{U}|\rv{Q},\rv{S}}, \label{eq_TV conditioning} \\
	\BP{\rv{U}|\rv{Q}} &\geq \BP{\rv{U}|\rv{Q},\rv{S}}, \label{eq_BP conditioning}\\ 
	\ENT{\rv{U}|\rv{Q}} &\geq \ENT{\rv{U}|\rv{Q},\rv{S}}. \label{eq_ENT conditioning}
\end{align}
\end{subequations}
\end{lemma}

\begin{remark}\label{rem_non binary extension}
	In this paper, we assume for simplicity that $\rv{U}$ is binary. It is possible to extend our results to the non-binary case. To this end, a suitable extension of the distribution parameters is required. The key properties that need to be preserved are (a) that they be bounded between $0$ and $1$; (b) that they approach their extreme values in tandem; and (c) that they satisfy \Cref{lem_effect of conditioning}. 
	In Appendix~\ref{app_non binary extension} we suggest a suitable extension that satisfies these requirements. 
\end{remark}

\subsection{Polarization}\label{sec_polarization}
We review some basics of polarization in this section. The concepts introduced here will be useful in the sequel.

\subsubsection{General Definitions}
Consider a strictly stationary process $(\rv{X}_j,\rv{Y}_j)$, $j=1,2,\ldots$ with a known joint distribution. We assume that $\rv{X}_j$ are binary and $\rv{Y}_j \in \mathcal{Y}$, where $\mathcal{Y}$ is a finite alphabet. The random variables $\rv{X}_j$ are to be estimated from the observations $\rv{Y}_j$. In a channel coding setting, $\rv{X}_j$ is the input to a channel and $\rv{Y}_j$ its output. In a lossless source coding setting~\cite{arikan_2010_source}, $\rv{X}_j$ is a data sequence to be compressed and $\rv{Y}_j$ is side information available to the decompressor. In a lossy compression setting~\cite{Korada_source}, the compressor takes a source sequence and distorts it to obtain a sequence $\rv{X}_1^N$ that is ultimately recovered by the decompressor.\footnote{In fact, in a lossy compression setting, with side information known to both compressor and decompressor, the process is $(\rv{X}_j, \rv{Y}_j)$, where $\rv{Y}_j = (\rv{Y}'_j, \rv{Y}''_j)$. The random variables $\rv{Y}'$ are the sequence to be compressed and the random variables $\rv{Y}''$ are the side information.} 

We denote Ar\i{}kan's polarization matrix by $G_N = B_N G_2^{\otimes n}$, where $N = 2^n$, $B_N$ is the $N\times N$ bit-reversal matrix, and $G_2 = \begin{bmatrix} 1 & 0 \\ 1 & 1 \end{bmatrix}$. Recall that $G_N^{-1} = G_N$. Following~\cite{sasoglu_2016}, we define 
\begin{subequations} \label{eq_defs of UVQR}
\begin{align} 
	\rv{U}_1^N &= \rv{X}_1^N G_N, \label{eq_defs of UVQR U}\\
	\rv{V}_1^N &= \rv{X}_{N+1}^{2N} G_N, \label{eq_defs of UVQR V}\\ 
	\rv{Q}_i   &= (\rv{U}_{1}^{i-1}, \rv{Y}_1^N), \label{eq_defs of UVQR Q} \\ 
	\rv{R}_i   &= (\rv{V}_1^{i-1}, \rv{Y}_{N+1}^{2N}), \label{eq_defs of UVQR R} 
\end{align}
\end{subequations}
where $i=1,2,\ldots, N$. 

Due to the recursive nature of polar codes, the above equations will be key for passing from a block of length $N$ to a block of length $2N$. First, however, let us concentrate on a length-$N$ block. For such a block, equations~\eqref{eq_defs of UVQR U} and~\eqref{eq_defs of UVQR Q} are pertinent. Although we have described several different communication scenarios, they all share the same succinct description that follows.

A certain subset of indices $F \subset \{1,2,\ldots,N\}$ is preselected according to some rule; the set $F$ dictates the performance of the code. When encoding (compressing), one produces a sequence $\rv{U}_1^N$. The relationship between the sequence $\rv{U}_1^N$ and the sequence $\rv{X}_1^N$ is given by~\eqref{eq_defs of UVQR U}.  Then, $\rv{U}_F$ is made available  to the decoder.\footnote{Depending on the application, this can be done either explicitly, by shared randomness, or both.} The decoding (decompressing) operation is iterative. For $i=1,2,\ldots$, the decoder estimates $\rv{U}_i$ from $\rv{Q}_i$; it uses its previous estimates of $\rv{U}_1^{i-1}$ to form $\rv{Q}_i$. Whenever it encounters an index in $F$, it returns as its estimate the relevant value from $\rv{U}_F$.  After estimating $\rv{U}_1^N$, the decoder recovers $\rv{X}_1^N$ via~\eqref{eq_defs of UVQR U}. 

The polarization phenomenon is that for large enough $n$, the fraction of indices with moderate conditional entropy, $|\{i: \ENT{\rv{U}_i|\rv{Q}_i} \in (\epsilon, 1-\epsilon)\}|/N$, becomes negligibly small for any $\epsilon>0$. One approach~\cite{Arikan_2009,arikan_2010_source} to derive such results is probabilistic. Rather than counting the number of indices with moderate conditional entropy, a sequence of random variables $\rv{H}_n$, $n=1,2,\ldots$ is defined. The random variable $\rv{H}_n$ assumes the value $\ENT{\rv{U}_i|\rv{Q}_i}$, with $i$ selected uniformly from $\{1,2,\ldots,N\}$. Thus, the probability that $\rv{H}_n$ lies in a certain range equals the fraction of indices whose conditional entropies lie in this range. 

The recursive nature of the polarization transform is at the heart of the probabilistic approach. 
Concretely, let $\rv{B}_1,\rv{B}_2, \ldots$ be a sequence of independent and identically distributed (i.i.d.) Bernoulli-$1/2$ random variables.  
We set $i-1=(\rv{B}_1 \rv{B}_2 \cdots \rv{B}_n)_2$; indeed, $i$ assumes any value in $\{1,2,\ldots,N\}$ with equal probability. Define the random variables
\begin{equation}
\begin{split}
\rv{K}_n &= \TV{\rv{U}_i |\rv{U}_1^{i-1}, \rv{Y}_1^N} = \TV{\rv{U}_i | \rv{Q}_i}, \\ 
\rv{Z}_n &= \BP{\rv{U}_i |\rv{U}_1^{i-1}, \rv{Y}_1^N } = \BP{\rv{U}_i | \rv{Q}_i},\\
\rv{H}_n &= \ENT{\rv{U}_i |\rv{U}_1^{i-1}, \rv{Y}_1^N } = \ENT{\rv{U}_i | \rv{Q}_i} 	
\end{split} \label{eq_defs of Kn Zn Hn}
\end{equation}
whenever $(i-1)=(\rv{B}_1 \rv{B}_{2} \cdots \rv{B}_n)_2$. That is, they denote the relevant distribution parameters for a uniformly chosen index after $n$ polarization steps. 
We call $\rv{K}_n, \rv{Z}_n,$ and $\rv{H}_n$, $n=1,2,\ldots$ the \emph{total variation distance process}, the \emph{Bhattacharyya process}, and the \emph{conditional entropy process}, respectively. 

When passing from a length-$N$ block to a block of length $2N$, by the properties of $G_N$~\cite[Section VII]{Arikan_2009}, 
\begin{equation} \rv{K}_{n+1} = \begin{cases}
 	\TV{\rv{U}_i + \rv{V}_i | \rv{Q}_i, \rv{R}_i} & \text{if } \rv{B}_{n+1} = 0 \\ 	
 	\TV{\rv{V}_i | \rv{U}_i + \rv{V}_i, \rv{Q}_i, \rv{R}_i} & \text{if } \rv{B}_{n+1} = 1.
 \end{cases}\label{eq_single step polarization for K} \end{equation}
Similar relationships hold for $\rv{H}_{n+1}$ and $\rv{Z}_{n+1}$. 
We shall use the mnemonics $\rv{K}_n^-$ and $\rv{K}_n^+$ to denote $\TV{\rv{U}_i + \rv{V}_i | \rv{Q}_i, \rv{R}_i}$ and $\TV{\rv{V}_i | \rv{U}_i+\rv{V}_i, \rv{Q}_i, \rv{R}_i}$, respectively. I.e., $\rv{K}_{n+1}$ assumes the value $\rv{K}_n^-$ when $\rv{B}_{n+1} = 0$ and the value $\rv{K}_n^+$ when $\rv{B}_{n+1} = 1$. We shall use similar mnemonics for $\rv{H}_n$ and $\rv{Z}_n$.

The probability law of  $(\rv{U}_i,\rv{V}_i,\rv{Q}_i,\rv{R}_i)$ can be obtained from the probability law of $(\rv{X}_1^{2N},\rv{Y}_1^{2N})$ using~\eqref{eq_defs of UVQR}. Moreover, for fixed $i$, there exists a function $f$, which depends solely on $i$, such that
\begin{equation}
\begin{split}
(\rv{U}_i,\rv{Q}_i) &= f(\rv{X}_1^N,\rv{Y}_1^N), \\ 
(\rv{V}_i,\rv{R}_i) &= f(\rv{X}_{N+1}^{2N},\rv{Y}_{N+1}^{2N}).
\end{split}
\label{eq_P(UVQR) from P(XY)}
\end{equation}
This can be seen by comparing \eqref{eq_defs of UVQR U} and \eqref{eq_defs of UVQR Q} with \eqref{eq_defs of UVQR V} and \eqref{eq_defs of UVQR R}. Due to stationarity, $\prrv{\rv{U}_i,\rv{Q}_i}{} = \prrv{\rv{V}_i,\rv{R}_i}{}$.

Denote $\rv{T}_i = \rv{U}_i +\rv{V}_i$, as in \Cref{fig_illustration of a polarization transform}. The mapping $(\rv{U}_i,\rv{V}_i) \mapsto (\rv{T}_i,\rv{V}_i)$ is one-to-one and onto. Hence, 
\begin{equation} \prrv{\rv{T}_i,\rv{V}_i,\rv{Q}_i,\rv{R}_i}{t,v,q,r} = \prrv{\rv{U}_i,\rv{V}_i,\rv{Q}_i,\rv{R}_i}{t+v,v,q,r}.\label{eq_P(STQR) from P(UVQR)}
\end{equation}

We  now   formally define polarization and fast polarization. 
\begin{definition}
Let $\rv{A}_n$, $n=1,2,\ldots$ be a sequence of random variables that take values in $[0,1]$. 
\begin{enumerate}
	\item The sequence $\rv{A}_n$ \emph{polarizes} if it converges almost surely to a $\{0,1\}$-random variable $\rv{A}_{\infty}$ as $n\to\infty$. We will sometimes abbreviate this by saying that ``$\rv{A}_n$ polarizes to $\rv{A}_{\infty}$.'' 	
	\item The sequence $\rv{A}_n$ \emph{polarizes fast to $0$} with $\beta >0$ if it polarizes and  
	\[\lim_{n\to\infty} \Prob{\rv{A}_n < 2^{-2^{n\beta}}} = \Prob{\rv{A}_{\infty} = 0}.\] 
	\item The sequence $\rv{A}_n$ \emph{polarizes fast to $1$} with $\beta >0$ if it polarizes and  
	\[\lim_{n\to\infty} \Prob{\rv{A}_n > 1 - 2^{-2^{n\beta}}} = \Prob{\rv{A}_{\infty} = 1}.\]
\end{enumerate}	
When the precise value of $\beta$ is either obvious from the context or not needed, we will write that $\rv{A}_n$ polarizes fast to, say, $0$, without mentioning the value of $\beta$. 
\end{definition}

The following lemma, first obtained by Ar\i{}kan and Telatar in~\cite{Arikan_Telatar_2009} and later adapted to the general case by \c{S}a\c{s}o\u{g}lu in~\cite{sasoglu_thesis}, is an important tool for establishing fast polarization for a sequence of random variables.  
\begin{lemma}\label{lem_simple proof}\cite{Arikan_Telatar_2009},\cite[Lemma 4.2]{sasoglu_thesis} 
Let $\rv{B}_n$, $n=1,2,\ldots$ be an i.i.d. Bernoulli-$1/2$ process and $\rv{A}_n$, $n=1,2,\ldots$ be a $[0,1]$-valued process that polarizes to a $\{0,1\}$-random variable $\rv{A}_{\infty}$. Assume that there exist $k \geq 1$ and $d_0,d_1 > 0$ such that for $i=0,1$, \[ \rv{A}_{n+1} \leq k \rv{A}_n^{d_i}\quad \text{if } \rv{B}_{n+1} = i.\] Then, for any $0 < \beta < E = (\log_2d_0 + \log_2 d_1)/2$, we have 
\begin{equation} 
\lim_{n\to\infty} \Prob{\rv{A}_n < 2^{-2^{n\beta}}} = \Prob{\rv{A}_{\infty} = 0}. \label{eq_simple proof result}
\end{equation}
\end{lemma}
\begin{remark}
It was shown in~\cite{Tal_2017_simple} that~\Cref{lem_simple proof} can be strengthened. Namely,  equation~\eqref{eq_simple proof result} can be replaced with the stronger assertion $\lim_{n_0 \to \infty} \Probi{\rv{A}_n \leq 2^{-2^{n\beta}} \; \text{for all } n \geq n_0} = \Probi{\rv{A}_{\infty}=0}$. Hence, any result based on \Cref{lem_simple proof}, such as \Cref{th_sasoglu_tal_2016,thm_fast polarization of Z to 1}, can be strengthened similarly. 
\end{remark}

\subsubsection{The Memoryless Case}
The memoryless case is characterized by $\prrv{\rv{X}_1^{N},\rv{Y}_1^{N}}{x_1^{N},y_1^N} = \prod_{j=1}^N \prrv{\rv{X}, \rv{Y}}{x_j,y_j}$.  
Ar\i{}kan showed in~\cite{Arikan_2009} that in the memoryless case the process $\rv{H}_n$ polarizes. Consequently, when $n$ is large enough, for all but a negligible fraction of indices $i$, $\ENT{\rv{U}_i|\rv{U}_1^{i-1},\rv{Y}_1^N}$ is either very close to $0$ or very close to $1$.  

To achieve this, Ar\i{}kan had shown that the sequence $\rv{H}_n$, $n=1,2,\ldots$ is a bounded martingale sequence and thus converges almost surely to some random variable $\rv{H}_{\infty}$. By showing that $\rv{H}_{\infty}$ can only assume the values $0$ and $1$, polarization is obtained. 

The Bhattacharyya process, in the memoryless case, is a bounded supermartingale that converges almost surely to a $\{0,1\}$-random variable $\rv{Z}_{\infty}$. The process $\rv{Z}_n$ satisfies \Cref{lem_simple proof} with $E=1/2$ by virtue of~\cite[Proposition 5]{Arikan_2009}, by which
	\[
		\rv{Z}_{n+1} = \begin{cases} \leq 2\rv{Z}_n & \text{if } \rv{B}_{n+1} = 0 \\ 
 									 \rv{Z}_n^2  & \text{if } \rv{B}_{n+1} = 1. 
 									 \end{cases}
	\]
	Thus, the Bhattacharyya process polarizes fast to $0$ with any $\beta < 1/2$. 
 
Fast polarization of the Bhattacharyya parameter is important for the performance analysis of polar codes. In particular, this was instrumental in Ar\i{}kan's proof that polar codes are capacity-achieving for binary-input,  memoryless, symmetric, channels~\cite{Arikan_2009}. Ar\i{}kan had upper-bounded the probability of error of polar codes by the union-Bhattacharyya bound. Thanks to fast polarization of the Bhattacharyya process to $0$, the bound converges to $0$.

The additional requirement of fast polarization of $\rv{Z}_n$ to $1$ is important for many applications of polar codes. For example, it is integral to source coding applications~\cite{Korada_source} and to channel coding without symmetry assumptions~\cite{Honda_Yamamoto_2013}.  
In~\cite[Theorem 16]{Korada_source}, 
this fast polarization  was established by showing that the process $\tilde{\rv{Z}}_n = 1-\rv{Z}_n^2$ polarizes fast to $0$ with $\beta<1/2$. Another way to see this, which we  pursue in the sequel, is via the total variation process $\rv{K}_n$. 

A consequence of \Cref{lem_TV distance bounds} is that if $\rv{K}_n$ polarizes fast to $0$ then $\rv{Z}_n$ must polarize fast to $1$. The total variation process $\rv{K}_n$ can be shown to polarize (we show this in \Cref{cor_Zn Kn converge} for a more general setting). Fast polarization of $\rv{K}_n$ to $0$ is obtained from \Cref{lem_simple proof} and the following proposition. 
	
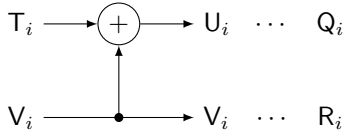
\begin{figure}
\begin{center}
\begin{tikzpicture}[>=latex]
	\node at(0,0) (CH1) {$ \cdots \quad \rv{Q}_i$};
	\node[ below = 0.7 of CH1] (CH2) {$ \cdots \quad \rv{R}_i$};
	\node[circle, draw, left = 1.4 of CH1, inner sep = 2pt] (plus) {$+$}; 
	\draw[<-] (plus) -- +(-1,0) node[left] (Si){$\rv{T}_i$}; 
	\draw[->] (plus) -- +(1,0) node(X){} node[right] {$\rv{U}_i$};
	\draw[->] (Si.east|-CH2) node[left] (Ti) {$\rv{V}_i$} -- (X|-Ti) node[right]{$\rv{V}_i$}  ;
	\draw[->] (Ti-|plus) -- (plus);
	\node[circle, fill, draw, minimum size = 3pt, inner sep = 0] at (Ti-|plus) {};   
\end{tikzpicture}
	\end{center}
	\caption{Illustration of a polarization transform. Random variables $(\rv{U}_i,\rv{Q}_i)$ have joint distribution $\prrv{\rv{U}_i,\rv{Q}_i}{}$ and random variables $(\rv{V}_i,\rv{R}_i)$ have joint distribution $\prrv{\rv{V}_i,\rv{R}_i}{}$. }\label{fig_illustration of a polarization transform} 
\end{figure}

\begin{proposition} \label{prop_K is a supermartingale} 
	Assume that $(\rv{X}_j,\rv{Y}_j)$, $j\in \mathbb{Z}$ is a memoryless process, where $\rv{X}_j$ is binary and $\rv{Y}_j \in \mathcal{Y}$.  
	Then, 
	\begin{equation}
		\rv{K}_{n+1} = \begin{cases} \rv{K}_n^2 & \text{if } \rv{B}_{n+1} = 0 \\ 
 									 \leq 2\rv{K}_n  & \text{if } \rv{B}_{n+1} = 1. 
 									 \end{cases}	\label{eq_polarization bounds for TV}
	\end{equation} 
\end{proposition}
In the sequel, we shall generalize this proposition to a non-memoryless case. The proof for the memoryless case serves as preparation for the more general case, which uses similar techniques. For an extension of \Cref{prop_K is a supermartingale} to the case where $\rv{X}_j$ is non-binary, see Appendix~\ref{app_non binary extension}. 
\begin{IEEEproof}
Fix $\rv{B}_1,\ldots, \rv{B}_{n}$ and let $i-1 = (\rv{B}_1 \rv{B}_{2} \cdots \rv{B}_n)_2$.   
This also fixes the value of $\rv{K}_n$.  
Using~\eqref{eq_P(UVQR) from P(XY)} and the memoryless assumption, we denote $P \equiv \prrv{\rv{U}_i,\rv{Q}_i}{} =  \prrv{\rv{V}_i,\rv{R}_i}{}$, by which  
 \[\prrv{\rv{U}_i,\rv{V}_i,\rv{Q}_i,\rv{R}_i}{u,v,q,r} = P(u,q)P(v,r).\]  
Note that $\rv{K}_n = \TV{\rv{U}_i|\rv{Q}_i} = \TV{\rv{V}_i|\rv{R}_i}$.
 
 Set $\rv{T}_i = \rv{U}_i+\rv{V}_i$; by~\eqref{eq_P(STQR) from P(UVQR)},  
 \[  \prrv{\rv{T}_i,\rv{V}_i,\rv{Q}_i,\rv{R}_i}{t,v,q,r} = P(t+v,q)P(v,r), \]
and $\prrv{\rv{T}_i,\rv{Q}_i,\rv{R}_i}{t,q,r} = \sum_{v=0}^1 \prrv{\rv{T}_i,\rv{V}_i,\rv{Q}_i,\rv{R}_i}{t,v,q,r}$.
 A single-step polarization from $\rv{K}_n$ to $\rv{K}_{n+1}$, \eqref{eq_single step polarization for K}, becomes
\begin{equation} \rv{K}_{n+1} = \begin{cases} \TV{\rv{T}_i|\rv{Q},\rv{R}_i} & \text{if } \rv{B}_{n+1} = 0 \\
 	\TV{\rv{V}_i|\rv{T}_i,\rv{Q}_i,\rv{R}_i} & \text{if } \rv{B}_{n+1} = 1. 
 \end{cases} \label{eq_Kn single step polarization}
 \end{equation}

Assume first that $\rv{B}_{n+1} = 0$. Then 
\begin{align*}
\rv{K}_{n+1} &= \sol{\sum_{q,r}} \left| \prrv{\rv{T}_i,\rv{Q}_i,\rv{R}_i}{0,q,r} - \prrv{\rv{T}_i,\rv{Q}_i,\rv{R}_i}{1,q,r} \right| \\ 
&= \sol{\sum_{q,r}} \left| \sum_{v=0}^1 P(v,r) (P(v,q)  - P(v+1,q))\right|	\\
&= \sol{\sum_{q,r}} \bigg|\Big(P(0,q)-P(1,q)\Big)\Big(P(0,r)-P(1,r)\Big)\bigg| \\ 
&\eqann{a} \sol{\sum_{q,r}} \left|P(0,q)-P(1,q)\right|\cdot\left|P(0,r)-P(1,r)\right| \\ 
&= \sum_q \left|P(0,q)-P(1,q)\right| \cdot \sum_r \left|P(0,r)-P(1,r)\right| \\ 
&= \rv{K}_n^2, 
\end{align*}
where \eqannref{a} is because $|ab|=|a|\cdot|b|$ for any two numbers $a,b$. 
Next, assume that $\rv{B}_{n+1} = 1$. Observe that for any four numbers $a,b,c,d$, 
\begin{equation} (ab-cd) = \frac{(a+c)(b-d)+(b+d)(a-c)}{2}.\label{eq_abcd equality} \end{equation}
With a slight abuse of notation, we denote $P(q) = \prrv{\rv{Q}_i}{q} = P(0,q) + P(1,q)$. Then,  $P(r) = \prrv{\rv{R}_i}{r}=P(0,r) + P(1,r)$. Thus, 
\begin{align*}
\rv{K}_{n+1} &= \sol{\sum_{t,q,r}} \left| \prrv{\rv{T}_i,\rv{V}_i,\rv{Q}_i,\rv{R}_i}{t,0,q,r} - \prrv{\rv{T}_i,\rv{V}_i,\rv{Q}_i,\rv{R}_i}{t,1,q,r} \right| \\ 
&= \sol{\sum_{t,q,r}} \left| P(t,q)P(0,r) - P(t+1,q)P(1,r)\right| \\ 
&\leq \frac{1}{2} \sol{\sum_{t,q,r}}P(q)\left| P(0,r) -  P(1,r)\right|\\ &\quad + \frac{1}{2}  \sol{\sum_{t,q,r}}P(r)\left| P(t,q) -  P(t+1,q)\right| 	\\
&= \frac{1}{2} \sol{\sum_{t,r}}\left| P(0,r) -  P(1,r)\right|\\ &\quad + \frac{1}{2}  \sol{\sum_{t,q}}\left| P(t,q) -  P(t+1,q)\right| \\
&= 2\rv{K}_n, 
\end{align*}
where the inequality is due to a combination of~\eqref{eq_abcd equality} with the triangle inequality.

We have shown that $\rv{K}_{n+1} = \rv{K}_n^2$ if $\rv{B}_{n+1} = 0$ and $\rv{K}_{n+1} \leq 2\rv{K}_n$ if $\rv{B}_{n+1} = 1$, completing the proof.
\end{IEEEproof} 

\begin{remark}
Several other authors have independently looked at the polarization of the total variation distance. For example,~\cite[Proposition 5.1]{alsan2015} derives relations similar to~\eqref{eq_polarization bounds for TV}; the top equality of~\eqref{eq_polarization bounds for TV} is also shown in~\cite[Equation 12]{Dumer_stepped}.  
Those results were derived for binary-input, memoryless, and symmetric channels. 
Our \Cref{prop_K is a supermartingale}, on the other hand, does not require symmetry. We note in passing that it is also easily extended to a non-stationary case (similar to~\cite[Appendix 2.A]{korada2009} for the Bhattacharyya process), but that is outside the scope of this paper. 
\end{remark}

\section{Finite-State Aperiodic Irreducible Markov Processes} \label{sec_FAIM}
In this section we introduce a class of processes with memory that we call \emph{ Finite-state Aperiodic Irrecducible Markov processes} (FAIM processes). This is the class of processes for which we establish polarization and fast polarization. 

These processes are described using an underlying state sequence. Often, however, the state sequence is hidden. The polarization results we obtain apply to processes with a hidden state sequence. 

\subsection{Definition}

Let $(\rv{X}_j,\rv{Y}_j, \rv{S}_j)$, $j\in \mathbb{Z}$ be a strictly stationary process, where $\rv{X}_j$ is binary, $\rv{Y}_j \in \mathcal{Y}$, and $\rv{S}_j \in \mathcal{S}$. The alphabets $\mathcal{Y}$ and $\mathcal{S}$ are finite; in particular, $\mathcal{S}=\{1,2,\ldots,|\mathcal{S}|\}$. 
We call $\rv{S}_j, j \in \mathbb{Z}$ the \emph{state sequence}; it governs the distribution of sequences $\rv{X}_j$ and $\rv{Y}_j$, $j \in \mathbb{Z}$.

We may think of $\rv{X}_j$ as a state-dependent input to a state-dependent channel with output $\rv{Y}_j$. Alternatively, $\rv{X}_j$ may be some state-dependent source to be compressed, and $\rv{Y}_j$ an observation that the decoder may use as a decompression aid. The state sequence encompasses the memory of the process.  

The process is described by the conditional probability $\prrv{\rv{X}_j,\rv{Y}_j,\rv{S}_j|\rv{S}_{j-1}}{}$, which, by the stationarity assumption, is independent of $j$. We assume a Markov property: conditioned on $\rv{S}_{j-1}$, the random variables $\rv{X}_k,\rv{Y}_k,\rv{S}_k$ are independent of $\rv{X}_l,\rv{Y}_l, \rv{S}_{l-1}$ for any $l < j \leq k$. Thus, for any $N>M>0$, 
\begin{equation} \label{eq_recursive computation of joint prob}
\begin{split}
	& \prrv{\rv{X}_1^{N},\rv{Y}_1^{N}, \rv{S}_{N} | \rv{S}_0}{} \\
	& \quad = \sum_{\mstate} \prrv{\rv{X}_1^{M},\rv{Y}_1^{M}, \rv{S}_{M}, \rv{X}_{M+1}^N,\rv{Y}_{M+1}^N, \rv{S}_N | \rv{S}_0}{} \\ 
	& \quad = \sum_{\mstate} \prrv{\rv{X}_{M+1}^N,\rv{Y}_{M+1}^N, \rv{S}_{N}|\rv{S}_M ,\rv{X}_1^{M},\rv{Y}_1^{M}, \rv{S}_0}{} \cdot \prrv{\rv{X}_1^{M},\rv{Y}_1^{M},\rv{S}_M | \rv{S}_0}{} \\
	& \quad = \sum_{\mstate} \prrv{\rv{X}_{M+1}^N,\rv{Y}_{M+1}^N, \rv{S}_{N}|\rv{S}_M}{} \cdot \prrv{\rv{X}_1^{M},\rv{Y}_1^{M},\rv{S}_M | \rv{S}_0}{}, 
\end{split} 
\end{equation}
where $\mstate$ in the sum represents the value of the middle state $\rv{S}_M$.  

The state sequence is a finite-state homogeneous Markov chain. We denote its marginal distribution by $\pi$, and use the  shorthand
\begin{equation}
\begin{split}
	\pi_N(\istate) &= \prrv{\rv{S}_N}{\istate} \\ 	
	\pi_{N|M}(\mstate|\istate) &= \prrv{\rv{S}_N|\rv{S}_M}{\mstate|\istate} \\
	\pi_{N,M}(\mstate,\istate) &= \prrv{\rv{S}_N,\rv{S}_M}{\mstate,\istate}, 
\end{split} \label{eq_def of Q0 Qmn}
\end{equation}
where $N>M$. 
Note that $\pi_N(\istate) = \pi_0(\istate)$ and $\pi_{N|M}(\mstate|\istate) = \pi_{N-M|0}(\mstate|\istate)$. 

A finite-state homogeneous Markov chain is aperiodic and irreducible (ergodic) if and only if there
is some $N_0 > 0$ such that for any $N\geq N_0$, $\pi_{N|0}(\mstate|\istate) >0$ for any
$\istate,\mstate \in \mathcal{S}$. It can be shown that it has a unique stationary distribution
$\pi_0$  and $\pi_0(\istate)>0$ for any $\istate \in \mathcal{S}$. Moreover,
$\pi_{N|0}(\mstate|\istate) \to \pi_0(\mstate)$ exponentially fast as $N\to \infty$ for any
$\istate,\mstate \in \mathcal{S}$. See, e.g.,~\cite[Section 8]{billingsley1995probability}.

The process $(\rv{X}_j,\rv{Y}_j, \rv{S}_j)$, $j\in \mathbb{Z}$ is called a \emph{finite-state aperiodic irreducible Markov} process if the underlying Markov process $\rv{S}_j, j \in \mathbb{Z}$ is homogenous, finite-state, strictly stationary, aperiodic, and irreducible.\footnote{We remark that the process $(\rv{X}_j,\rv{Y}_j)$, $j \in \mathbb{Z}$ is not necessarily Markov.}  
In the sequel, we assume that $(\rv{X}_j,\rv{Y}_j, \rv{S}_j)$, $j\in \mathbb{Z}$ is a FAIM process.

At this point, the reader may wonder why we have imposed aperiodicity and irreducibility. In~\cite[Theorem 3]{sasoglu_2016}, it was demonstrated that periodic processes may not polarize. We assume aperiodicity to ensure that polarization indeed happens.
As for irreducibility, note that since the number of states is finite, the state sequence $\rv{S}_j, j \in \mathbb{Z}$ must reach an irreducible sink after sufficient time. Hence, the irreducibility assumption is equivalent to assuming that the state sequence begins in some irreducible sink.

Our model applies to many problems in information theory that can be described using states. For
example, compression of finite memory sources and coding for input constrained channels.
Additionally, our model may be applied to finite-state channels; in this case, the FAIM state
sequence describes both the channel state and input state. That is, FAIM processes enable us to
model non-i.i.d. input sequences. 

One famous example of a finite state model is the indecomposable FSC model considered in \cite[Section 4.6]{Gallager}. There are some differences between this model and ours. Most importantly, a FAIM process has a specified input distribution, whereas an indecomposable FSC is devoid of such specification. Instead, an indecomposable FSC imposes conditions that should hold for all input sequences. That said, once a hidden Markov input distribution has been specified, we can define a process in which the state space is the Cartesian product of the state spaces of the input distribution and the channel. In many important cases, e.g.\ a Gilbert-Elliot channel~\cite{mushkin_ge_channel}, this combined process falls under the FAIM framework.

\subsection{Blocks of a FAIM Process}
Typically, the state sequence is not observed. The joint distribution of $(\rv{X}_1^N,\rv{Y}_1^N)$ is given by 
\[ \prrv{\rv{X}_1^N,\rv{Y}_1^N}{x_1^N,y_1^N} = \sum_{\mstate,\istate} \prrv{\rv{X}_1^N,\rv{Y}_1^N, \rv{S}_N|\rv{S}_0}{x_1^N,y_1^N,\mstate|\istate} \pi_0(\istate),\]
where $\pi_0$ is the stationary distribution of the initial state.  

\begin{definition}[Block]
	Let $(\rv{X}_j,\rv{Y}_j, \rv{S}_j)$, $j\in \mathbb{Z}$ be a FAIM process and assume $M>L$. We call $(\rv{X}_{L+1}^M, \rv{Y}_{L+1}^M)$ a \emph{block} of the FAIM process. Its length is $M-L$.
	
	State $\rv{S}_{L}$ is called the \emph{initial} state of the block. State $\rv{S}_M$ is called the \emph{final} state of the block.
\end{definition}
We emphasize that the initial state of the block $(\rv{X}_{L+1}^M, \rv{Y}_{L+1}^M)$ is $\rv{S}_L$ and \emph{not} $\rv{S}_{L+1}$.  

The following lemma holds for any two non-overlapping blocks of a FAIM process. It establishes that FAIM processes are a special case of the family of processes considered in~\cite{sasoglu_2016}. 
\begin{lemma}\label{lem_XY is psi-mixing}
Assume that 
$(\rv{X}_j,\rv{Y}_j, \rv{S}_j)$, $j\in \mathbb{Z}$ is a FAIM process. Then, there exists a non-increasing sequence $\psi(N)$, $\psi(N) \to 1$ as $N \to \infty$, such that for any $N > M \geq L \geq 1$, 
\begin{equation} \prrv{\rv{X}_1^L,\rv{Y}_1^L , \rv{X}_{M+1}^N, \rv{Y}_{M+1}^N}{} \leq \psi(M-L) \cdot \prrv{\rv{X}_1^L,\rv{Y}_1^L}{} \cdot \prrv{\rv{X}_{M+1}^N, \rv{Y}_{M+1}^N}{},\label{eq_psi mixing for xy} \end{equation}
and $\psi(0)<\infty$. 
\end{lemma}
We relegate the proof to Appedix~\ref{app_proof of XY is psi mixing}. We remark, however, that
\begin{equation}
\psi(N) = \begin{dcases} \max_{\istate,\mstate} \frac{\pi_{N|0}(\mstate|\istate)}{\pi_0(\mstate)} & \text{if } N>0 \\[0.1cm]
 \max_{\istate} \frac{1}{\pi_0(\istate)} &\text{if } N = 0. 	
 \end{dcases}
\label{eq_def of psi(n)} 	
\end{equation}
I.e., $\psi(\cdot)$ is  completely determined by the distribution of the underlying state sequence. Indeed, $\psi(N) \to 1$ as $N \to \infty$.  

A process satisfying~\eqref{eq_psi mixing for xy} with $\psi(N) \to 1$ as $N \to \infty$ is called \emph{$\psi$-mixing}.\footnote{In some literature, e.g.~\cite{bradley2007}, the term used is $\psi^*$-mixing.} The function $\psi(\cdot)$ is called the \emph{mixing coefficient}.  
The operational meaning of~\eqref{eq_psi mixing for xy} is that as $L$ and $M$ becomes more separated in time, the blocks  $(\rv{X}_1^L, \rv{Y}_1^L)$ and  $(\rv{X}_{M+1}^N, \rv{Y}_{M+1}^N)$ become almost independent.\footnote{Let $\mathcal{A}$ and $\mathcal{B}$ be two $\sigma$-algebras. If for any two events $A \in \mathcal{A}$ and $B \in \mathcal{B}$ we have $\Probi{A \cap B} \leq \Probi{A} \Probi{B}$ then $\Probi{A\cap B} = \Probi{A} \Probi{B}$. Assume to the contrary that for some events $A_0,B_0$, $\Probi{A_0 \cap B_0} < \Probi{A_0} \Probi{B_0}$. Denote the complement of $A_0$ by $\bar{A}_0$. Since $\bar{A}_0 \in \mathcal{A}$, we obtain a contradiction: $\Probi{B_0} = \Probi{\bar{A}_0 \cap B_0} + \Probi{A_0 \cap B_0} < \Probi{A_0} \Probi{B_0} + \Probi{\bar{A}_0}\Probi{B_0} = \Probi{B_0}$.}

Two adjacent blocks of the process share a state. The final state of the first block is the initial state of the second block. Given the shared state, the two blocks are independent. We capture this in the following lemma. 
\begin{lemma}\label{lem_two adjacent blocks independent given shared state}
For any $N>M\geq 1$, 
\begin{subequations} \label{eq_block independence given state}
\begin{align}	
\prrv{\rv{X}_1^M,\rv{Y}_1^M, \rv{X}_{M+1}^N,\rv{Y}_{M+1}^N|\rv{S}_M}{} &= \prrv{\rv{X}_1^M,\rv{Y}_1^M|\rv{S}_M}{} \prrv{\rv{X}_{M+1}^N,\rv{Y}_{M+1}^N|\rv{S}_M}{}, \label{eq_block independence given state N}\\
 \prrv{\rv{X}_1^M,\rv{Y}_1^M, \rv{X}_{M+1}^N,\rv{Y}_{M+1}^N|\rv{S}_0,\rv{S}_M,\rv{S}_N}{} &= \prrv{\rv{X}_1^M,\rv{Y}_1^M|\rv{S}_0,\rv{S}_M}{} \prrv{\rv{X}_{M+1}^N,\rv{Y}_{M+1}^N|\rv{S}_M,\rv{S}_N}{}. \label{eq_block independence given state 0 N M}
 \end{align}
\end{subequations}
\end{lemma}
This is a direct consequence of the Markov property. A formal derivation can be found in Appendix~\ref{app_proof of XY is psi mixing}.

A notational convention concludes this section. Our analysis involves the use of some states of blocks of a FAIM process. We will use ascending letters to denote values of ordered states. That is, a state with value $\istate$ occurs before a state with value $\mstate$, which, in turn, occurs before a state with value $\fstate$. In \Cref{fig_two blocks} we illustrate a particular case that will be used in the sequel. A block of length $2N$ comprises two adjacent blocks of length $N$. State $\rv{S}_0$, the initial state of the first block, may take value $\istate$, state $\rv{S}_N$, at the end of the first block and the beginning of the second block, may take value $\mstate$, and state $\rv{S}_{2N}$, at the end of the second block, may take value $\fstate$. 	We emphasize that $\istate,\mstate,\fstate \in \mathcal{S}$ are \emph{not} random variables, but \emph{values} of the relevant states.

\begin{figure}
\begin{center}
\begin{tikzpicture}
	\node[draw, thick, minimum width = 4.1cm, minimum height = 0.9 cm] (box1) at (0,0) {$(\rv{U}_i,\rv{Q}_i) = f(\rv{X}_1^N, \rv{Y}_1^N)$}; 
	\node[draw, thick, minimum width = 4.1cm, minimum height = 0.9 cm, right = -1pt of box1] (box2) {$(\rv{V}_i,\rv{R}_i) = f(\rv{X}_{N+1}^{2N},\rv{Y}_{N+1}^{2N})$}; 
	\node[above = 0.1 of box1.north west] {$\rv{S}_0$};
	\node[below = 0.4 of box1.south west, anchor = base] {$\istate$};
	\node[above = 0.1 of box1.north east] {$\rv{S}_{N}$};
	\node[below = 0.4 of box1.south east, anchor = base] {$\mstate$};
	\node[above = 0.1 of box2.north east] {$\rv{S}_{2N}$};
	\node[below = 0.4 of box2.south east, anchor = base] {$\fstate$};
\end{tikzpicture}
\end{center}
\caption{Two adjacent length-$N$ blocks of a FAIM process.  When $i-1 = (\rv{B}_1\rv{B}_2\cdots\rv{B}_n)_2$, there is a function $f$ such that $(\rv{U}_i,\rv{Q}_i) = f(\rv{X}_1^N,\rv{Y}_1^N)$ and $(\rv{V}_i,\rv{R}_i) = f(\rv{X}_{N+1}^{2N}, \rv{Y}_{N+1}^{2N})$. The initial state of the first block, $\rv{S}_0$, assumes value $\istate \in\mathcal{S}$. The final state of the first block, $\rv{S}_N$, which is also the initial state of the second block, assumes value $\mstate \in \mathcal{S}$. The final state of the second block, $\rv{S}_{2N}$, assumes value $\fstate \in \mathcal{S}$.}\label{fig_two blocks}
\end{figure}
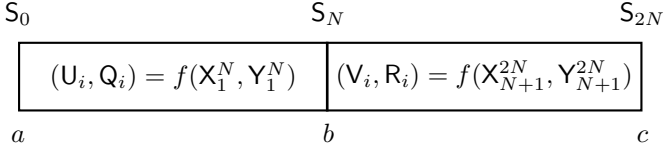

\subsection{Boundary-State-Informed Parameters for FAIM Processes}
Let  $(\rv{X}_1^N,\rv{Y}_1^N)$ be a block of a FAIM process with state sequence $\rv{S}_j$.  
Let $f(\cdot, \cdot)$ be some function independent of the state sequence such that
 \[(\rv{U},\rv{Q}) = f(\rv{X}_1^N, \rv{Y}_1^N)\] 
and $\rv{U}$ is binary.  
We denote   
\begin{equation}  P_{\istate}^{\mstate}(u,q) \triangleq \prrv{\rv{U},\rv{Q}|\rv{S}_N,\rv{S}_0}{u,q|\mstate,\istate} = \frac{\prrv{\rv{U},\rv{Q},\rv{S}_N|\rv{S}_0}{u,q,\mstate|\istate}}{\pi_{N|0}(\mstate|\istate)}. \label{eq_def of Pabuq}\end{equation} 
I.e., this is the distribution of $\rv{U}$ and $\rv{Q}$, functions of a block of length $N$, conditioned on the initial state being $\rv{S}_0 = \istate$ and the final state being $\rv{S}_N = \mstate$.  We further define 
\begin{equation}
	P_{\istate}^{\mstate}(q) = P_{\istate}^{\mstate}(0,q) + P_{\istate}^{\mstate}(1,q). \label{eq_def of Pabq}	
\end{equation}
I.e., $P_{\istate}^{\mstate}(q) = \prrv{\rv{Q}|\rv{S}_N,\rv{S}_0}{q|\mstate,\istate}$, and $\sum_q P_{\istate}^{\mstate}(q) = 1$.

We denote the results of replacing $\pr{U,Q}{u,q}$ with $P_{\istate}^{\mstate}(u,q)$ in \cref{eq_def of Z,eq_def of TV,eq_def of condent} by $\BPSn{\istate}{\mstate}{\rv{U}|\rv{Q}}$, $\TVSn{\istate}{\mstate}{\rv{U}|\rv{Q}}$, and $\ENTSn{\istate}{\mstate}{\rv{U}|\rv{Q}}$, respectively. For example, 
\begin{equation} \TVSn{\istate}{\mstate}{\rv{U}|\rv{Q}} = \sum_q \left| P_{\istate}^{\mstate}(0,q)
- P_{\istate}^{\mstate}(1,q) \right|. \label{eq_def of Kab}\end{equation}
Since $\prrv{\rv{U},\rv{Q},\rv{S}_N,\rv{S}_0}{u,q,\mstate,\istate} = P_{\istate}^{\mstate}(u,q)\cdot \pi_{N,0}(\mstate,\istate)$,
 we have
$\TV{\rv{U}|\rv{Q},\rv{S}_N,\rv{S}_0} = \sum_{\istate, \mstate} \pi_{N,0}(\mstate,\istate) \TVSn{\istate}{\mstate}{\rv{U}|\rv{Q}}.$ 
This leads to the following definition. 
\begin{definition}
Let $(\rv{U},\rv{Q}) = f(\rv{X}_1^N,\rv{Y}_1^N)$ with $\rv{U}$ binary.
The \emph{boundary-state-informed} (BSI) total variation distance, Bhattacharyya parameter, and conditional entropy are respectively defined as 
	\begin{align*} 
	\TV{\rv{U}|\rv{Q},\rv{S}_N,\rv{S}_0} &= \sol{\sum_{\istate, \mstate}} \pi_{N,0}(\mstate,\istate) \TVSn{\istate}{\mstate}{\rv{U}|\rv{Q}},\\
	\BP{\rv{U}|\rv{Q},\rv{S}_N,\rv{S}_0} &=  \sum_{\istate, \mstate} \pi_{N,0}(\mstate,\istate) \BPSn{\istate}{\mstate}{\rv{U}|\rv{Q}},\\ 
 \ENT{\rv{U}|\rv{Q},\rv{S}_N,\rv{S}_0} &=  \sum_{\istate, \mstate} \pi_{N,0}(\mstate,\istate) \ENTSn{\istate}{\mstate}{\rv{U}|\rv{Q}}.  
	\end{align*}
\end{definition}
BSI parameters are defined for blocks of the process; they depend on the initial and final states of the  block. 
Invoking~\eqref{eq_conditioning reduces} we relate the distribution parameters to their BSI counterparts,  
\begin{equation} 
\begin{split}
	\TV{\rv{U}|\rv{Q}} &\leq \TV{\rv{U}|\rv{Q},\rv{S}_N,\rv{S}_0},\\ 
	\BP{\rv{U}|\rv{Q}} &\geq \BP{\rv{U}|\rv{Q},\rv{S}_N,\rv{S}_0},\\
	\ENT{\rv{U}|\rv{Q}} &\geq \ENT{\rv{U}|\rv{Q},\rv{S}_N,\rv{S}_0}.
	\end{split}\label{eq_upper bound on KUQ} 
\end{equation}

\section{Fast Polarization for FAIM Processes} \label{sec_fast polarization}
This section contains our main result: fast polarization for FAIM processes. First, we show that they polarize by leveraging the results of~\cite{sasoglu_2016}. Then, we show fast polarization of the Bhattacharyya parameter and of the total variation distance to zero. 

The notation of \Cref{sec_polarization} holds, without change, for FAIM processes. That is, $\rv{U}_1^N, \rv{V}_1^N, \rv{Q}_i, \rv{R}_i$, $i=1,\ldots, N$ are defined using~\eqref{eq_defs of UVQR}.  
The random variables $\rv{B}_1, \ldots, \rv{B}_n$ are used for a random, iterative, uniform selection of an index
after $n$ polarization steps. That is,  they constitute the binary expansion of $i-1$, through which the random variables $\rv{K}_n = \TV{\rv{U}_i|\rv{Q}_i}$, $\rv{H}_n = \ENT{\rv{U}_i|\rv{Q}_i}$, and $\rv{Z}_n = \BP{\rv{U}_i|\rv{Q}_i}$ are defined.  Random variable $\rv{K}_{n+1}$ is related to $\rv{K}_n$ by~\eqref{eq_single step polarization for K}. I.e., $\rv{K}_{n+1} = \rv{K}_n^{-}$ if $\rv{B}_{n+1} = 0$ and $\rv{K}_{n+1} = \rv{K}_n^{+}$ if $\rv{B}_{n+1} = 1$.  Similar relationships hold for $\rv{H}_n$ and $\rv{Z}_n$. 

Let $\hat{\rv{K}}_n, \hat{\rv{H}}_n,$ and $\hat{\rv{Z}}_n$ denote the boundary-state-informed versions of $\rv{K}_n, \rv{Z}_n$, and $\rv{H}_n$, respectively. That is, 
\begin{equation}
\begin{split}
	\hat{\rv{K}}_n &= \TV{\rv{U}_i|\rv{Q}_i,\rv{S}_N,\rv{S}_0}, \\ 
	\hat{\rv{Z}}_n &= \BP{\rv{U}_i|\rv{Q}_i,\rv{S}_N,\rv{S}_0}, \\
	\hat{\rv{H}}_n &= \ENT{\rv{U}_i|\rv{Q}_i,\rv{S}_N,\rv{S}_0},
\end{split} \label{eq_defs of Knhat Znhat Hnhat}
\end{equation}
where $i-1 = (\rv{B}_1\rv{B}_2\cdots\rv{B}_n)_2$. 
By~\eqref{eq_upper bound on KUQ}, $\rv{K}_n \leq \hat{\rv{K}}_n$, $\rv{Z}_n \geq \hat{\rv{Z}}_n$, and $\rv{H}_n \geq \hat{\rv{H}}_n$ for any $n$. Similar to~\eqref{eq_single step polarization for K}, we have
\begin{equation} \hat{\rv{K}}_{n+1} = \begin{cases}
 	\TV{\rv{U}_i + \rv{V}_i | \rv{Q}_i, \rv{R}_i,\rv{S}_0,\rv{S}_{2N}} & \text{if } \rv{B}_{n+1} = 0 \\ 	
 	\TV{\rv{V}_i | \rv{U}_i + \rv{V}_i, \rv{Q}_i, \rv{R}_i,\rv{S}_0,\rv{S}_{2N}} & \text{if } \rv{B}_{n+1} = 1.
 \end{cases}\label{eq_single step polarization for Khat} \end{equation}
 Relationships akin to~\eqref{eq_single step polarization for Khat} hold for $\hat{\rv{Z}}_{n+1}$ and $\hat{\rv{H}}_{n+1}$, with $\TV{}$ replaced with $\BP{}$ and $\ENT{}$, respectively. We use the mnemonic $\hat{\rv{K}}_{n+1}^{-} = \TV{\rv{U}_i + \rv{V}_i | \rv{Q}_i, \rv{R}_i,\rv{S}_0,\rv{S}_{2N}}$ and $\hat{\rv{K}}_{n+1}^{+} =\TV{\rv{V}_i | \rv{U}_i + \rv{V}_i, \rv{Q}_i, \rv{R}_i,\rv{S}_0,\rv{S}_{2N}}$, and similar mnemonics for the BSI Bhattachryya and conditional entropy processes.

\subsection{Existing Polarization Results for FAIM Processes}
In~\cite{sasoglu_2016}, a class of processes with memory was considered. For this class, the authors showed that the conditional entropy process polarizes and that the Bhattacharyya process polarizes fast to $0$. 

Specifically, let 
\[ \ENTi{\rv{X}|\rv{Y}} \triangleq \sol[l]{\lim_{N\to \infty}} \frac{1}{N}\ENT{\rv{X}_1^N|\rv{Y}_1^N}.\]
This limit exists due to stationarity~\cite[Section 4.2]{cover_thomas} and the identity $\ENT{\rv{X}_1^N|\rv{Y}_1^N} = \ENT{\rv{X}_1^N,\rv{Y}_1^N} - \ENT{\rv{Y}_1^N}$.  
\begin{theorem} \label{th_sasoglu_tal_2016}
\cite[Theorems 1,2,4,5]{sasoglu_2016} For a strictly stationary $\psi$-mixing process $(\rv{X}_j,\rv{Y}_j)$, $j\in \mathbb{Z}$, with $\psi(0)<\infty$: 
\begin{enumerate} 
\item $\rv{H}_n$ polarizes to $\rv{H}_{\infty}$ with $\Prob{\rv{H}_{\infty} = 1} = \ENTi{\rv{X}|\rv{Y}}$; 
\item $\rv{Z}_n$ polarizes fast to $0$ with $\beta < 1/2$.
\end{enumerate}

In particular, for any $\epsilon > 0$,
\begin{subequations}\label{eq_simple proof}
\begin{align}
\lim_{N\to \infty} \frac{1}{N} \left|\left\{i: \ENT{\rv{U}_i| \rv{Q}_i	} > 1-\epsilon \right\}\right| &= \ENTi{\rv{X}|\rv{Y}}, \label{eq_simple ent 1}\\
\lim_{N\to \infty} \frac{1}{N} \left|\left\{i: \ENT{\rv{U}_i| \rv{Q}_i	} < \epsilon\right\} \right| &= 1-\ENTi{\rv{X}|\rv{Y}}, \label{eq_simple ent 2}
\end{align} \end{subequations}
and for any $\beta < 1/2$, 
\begin{equation} \lim_{N\to \infty} \frac{1}{N} \left|\left\{i: \BP{\rv{U}_i| \rv{Q}_i	} < 2^{-N^{\beta}} \right\} \right| = 1-\ENTi{\rv{X}|\rv{Y}}. \label{eq_simple fast polarization of Zn to 0} \end{equation}
\end{theorem}

To prove \Cref{th_sasoglu_tal_2016}, the conditional entropy process $\rv{H}_n$ was shown to be a bounded supermartingale, so it converges almost surely to some random variable $\rv{H}_{\infty}$. This latter random variable was shown to  be a $\{0,1\}$-random variable with $\Probi{\rv{H}_{\infty} = 1} = 1- \Probi{\rv{H}_{\infty} = 0} = \ENTi{\rv{X}|\rv{Y}}$. This yields~\eqref{eq_simple proof}. 

\Cref{eq_simple fast polarization of Zn to 0} is based on the observation that
\begin{equation} \Prob{\rv{Z}_n < 2^{-N^{\beta}}} = \frac{1}{N} \left|\left\{i: \BP{\rv{U}_i| \rv{Q}_i	} < 2^{-N^{\beta}} \right\} \right|.\label{eq_alternative fast polarization of Zn to 0}\end{equation}
First, the Bhattacharyya process $\rv{Z}_n$ was also shown to converge almost surely to $\rv{H}_{\infty}$. 
Next, using the mixing property, the authors showed that $\rv{Z}_n^- \leq 2\psi(0) \rv{Z}_n$ and $\rv{Z}_n^+ \leq \psi(0) \rv{Z}_n^2$. This allowed them to invoke \Cref{lem_simple proof} and obtain~\eqref{eq_simple fast polarization of Zn to 0}.
 
\begin{corollary}\label{cor_polarization for FAIM from ST2016}
	Let 
$(\rv{X}_j,\rv{Y}_j, \rv{S}_j)$, $j\in \mathbb{Z}$ be a FAIM process. Then,
\begin{enumerate} 
\item Its conditional entropy process $\rv{H}_n$ polarizes to $\rv{H}_{\infty}$ with $\Prob{\rv{H}_{\infty} = 1} = \ENTi{\rv{X}|\rv{Y}}$. 
\item Its Bhattacharyya process $\rv{Z}_n$ polarizes fast to $0$ with any $\beta < 1/2$. 
\end{enumerate}
\end{corollary}
 \begin{IEEEproof}
 	By \Cref{lem_XY is psi-mixing}, blocks of FAIM processes are $\psi$-mixing and satisfy the requirements of \Cref{th_sasoglu_tal_2016}.
 \end{IEEEproof}

\Cref{th_sasoglu_tal_2016}, and consequently \Cref{cor_polarization for FAIM from ST2016}, are silent on the rate of polarization of $\rv{Z}_n$ to $1$.   
In the sequel we  establish a compatible claim for FAIM processes. To do this, we exploit the structure of FAIM processes by calling upon the BSI processes $\hat{\rv{H}}_n$ and $\hat{\rv{K}}_n$.

\subsection{Polarization of the BSI Distribution Parameters}
This section is concerned with proving that the BSI distribution parameters polarize. We achieve this by first showing that the BSI conditional entropy polarizes and then using \Cref{lem_TV distance bounds} to establish polarization of the BSI Bhattacharyya parameter and BSI total variation distance. 

\begin{theorem}\label{th_BSI entropy polarizes}
Let $(\rv{X}_j,\rv{Y}_j,\rv{S}_j)$, $j \in \mathbb{Z}$ be a FAIM process. The BSI conditional entropy process $\hat{\rv{H}}_n$ polarizes to  $ \hat{\rv{H}}_{\infty}$ and $\hat{\rv{H}}_{\infty} = \rv{H}_{\infty}$ almost surely. 

In particular, for any $\epsilon > 0$,
\begin{align*}
\lim_{N\to \infty} \frac{1}{N} \left|\left\{i: \ENT{\rv{U}_i| \rv{Q}_i, \rv{S}_0, \rv{S}_N} > 1-\epsilon \right\}\right| &= \ENTi{\rv{X}|\rv{Y}}, \\
\lim_{N\to \infty} \frac{1}{N} \left|\left\{i: \ENT{\rv{U}_i| \rv{Q}_i, \rv{S}_0, \rv{S}_N} < \epsilon\right\} \right| &= 1-\ENTi{\rv{X}|\rv{Y}}.
\end{align*}	
\end{theorem}

\begin{IEEEproof}
Consider two adjacent blocks of length $N=2^n$ and let $i-1 = (\rv{B}_{1}\rv{B}_{2}\cdots\rv{B}_n)_2$. Recall from~\eqref{eq_P(UVQR) from P(XY)} that $(\rv{U}_i,\rv{Q}_i) = f(\rv{X}_1^N, \rv{Y}_1^N)$ and $(\rv{V}_i,\rv{R}_i)=f(\rv{X}_{N+1}^{2N}, \rv{Y}_{N+1}^{2N})$, where the function $f$ depends on the index $i$ (see \Cref{fig_two blocks}). Using~\eqref{eq_block independence given state 0 N M} we obtain 
\begin{equation} \prrv{\rv{U}_i,\rv{V}_i | \rv{Q}_i,\rv{R}_i,\rv{S}_0,\rv{S}_N,\rv{S}_{2N}}{} = \prrv{\rv{U}_i| \rv{Q}_i,\rv{S}_0,\rv{S}_N}{}\prrv{\rv{V}_i|\rv{R}_i, \rv{S}_N,\rv{S}_{2N}}{}. \label{eq_Puvqrs0sNs2N as product} \end{equation}

Thus,
\begin{align*}
\hat{\rv{H}}_n
&\eqann{a} \frac{1}{2}\Big(\ENT{\rv{U}_i|\rv{Q}_i,\rv{S}_0, \rv{S}_N} + \ENT{\rv{V}_i|\rv{R}_i,\rv{S}_N,\rv{S}_{2N}}\Big) \\ 	
&\eqann{b} \frac{1}{2}\ENT{\rv{U}_i,\rv{V}_i|\rv{Q}_i,\rv{R}_i,\rv{S}_0,\rv{S}_N,\rv{S}_{2N}} \\
&\eqann{c} \frac{1}{2}\ENT{\rv{U}_i+\rv{V}_i,\rv{V}_i|\rv{Q}_i,\rv{R}_i,\rv{S}_0,\rv{S}_N,\rv{S}_{2N}} \\
&\eqann{d} \frac{1}{2}\Big(\ENT{\rv{U}_i+\rv{V}_i|\rv{Q}_i,\rv{R}_i,\rv{S}_0,\rv{S}_N,\rv{S}_{2N}} \\&\quad+ \ENT{\rv{V}_i|\rv{U}_i+\rv{V}_i,\rv{Q}_i,\rv{R}_i,\rv{S}_0,\rv{S}_N,\rv{S}_{2N}}\Big)\\
&\eqann[\leq]{e} \frac{1}{2}\Big(\ENT{\rv{U}_i+\rv{V}_i|\rv{Q}_i,\rv{R}_i,\rv{S}_0,\rv{S}_{2N}} \\&\quad+ \ENT{\rv{V}_i|\rv{U}_i+\rv{V}_i,\rv{Q}_i,\rv{R}_i,\rv{S}_0,\rv{S}_{2N}}\Big)\\
&=\frac{1}{2}\Big(\hat{\rv{H}}_n^{-} + \hat{\rv{H}}_n^{+}\Big), \\
\end{align*}
where \eqannref{a} is by stationarity, \eqannref{b} is by~\eqref{eq_Puvqrs0sNs2N as product}, \eqannref{c} is because the mapping $(\rv{U},\rv{V}) \mapsto (\rv{U}+\rv{V},\rv{V})$ is one-to-one and onto, \eqannref{d} is by the chain rule for entropies, and \eqannref{e} is by~\eqref{eq_ENT conditioning}.  

By~\eqref{eq_single step polarization for Khat} (applied to the BSI conditional entropy), $\hat{\rv{H}}_n$ is a submartingale sequence: 
\[\frac{1}{2}\Big(\hat{\rv{H}}_n^{-} + \hat{\rv{H}}_n^{+}\Big) =  \Exp{\hat{\rv{H}}_{n+1}\Big|\hat{\rv{H}}_n, \hat{\rv{H}}_{n-1},\ldots,\hat{\rv{H}}_1}  \geq \hat{\rv{H}}_n.\] %
It is also bounded, as $\hat{\rv{H}}_n \in [0,1]$ for any $n$. Thus, it converges almost surely to some random variable $\hat{\rv{H}}_{\infty} \in [0,1]$,~\cite[Theorem 35.4]{billingsley1995probability}. 

Denote $\Delta\rv{H}_n = \rv{H}_n -\hat{\rv{H}}_n$. The sequence $\Delta\rv{H}_n$ converges almost surely to the random variable $\Delta\rv{H}_{\infty} = \rv{H}_{\infty}-\hat{\rv{H}}_{\infty}$. This is because $\hat{\rv{H}}_n$ converges almost surely to $\hat{\rv{H}}_{\infty}$, and, by \Cref{cor_polarization for FAIM from ST2016}, $\rv{H}_n$ converges almost surely to $\rv{H}_{\infty}$.  
By~\eqref{eq_conditioning reduces}, $\Delta\rv{H}_n \geq 0$ for any $n$, which implies that $\Delta\rv{H}_{\infty} \geq 0 $ almost surely. We now show that $\Delta\rv{H}_{\infty}  = 0$ almost surely. To this end, we will need the following lemma, whose proof is postponed to the end of this theorem. 
\begin{lemma}\label{lem_DeltaH}
	The sequence $\Delta\rv{H}_n$ satisfies 
		\[ \lim_{n\to\infty} \Exp{\Delta\rv{H}_n} = 0.\] 
\end{lemma}

Since $\Delta\rv{H}_n$ converges to $\Delta \rv{H}_{\infty}$ almost surely, we specifically have
$\liminf_{n\to\infty} \Delta\rv{H}_n = \Delta\rv{H}_{\infty}$ almost surely. 
Using Fatou's lemma\footnote{Fatou's lemma~\cite[Theorem 16.3]{billingsley1995probability} states that if $\rv{A}_n$, $n=1,2,\ldots$ is a sequence of non-negative random variables then $\Exp{\liminf_{n\to\infty} \rv{A}_n} \leq \liminf_{n\to\infty}\Exp{\rv{A}_n}$.} for the non-negative sequence $\Delta\rv{H}_n$, $n=1,2,\ldots$ we obtain
\begin{align*} 0 \leq \Exp{\Delta\rv{H}_{\infty}} &=  \Exp{\liminf_{n\to\infty} \Delta\rv{H}_n}\\ &\leq 
\liminf_{n\to\infty}\Exp{\Delta\rv{H}_n} = \lim_{n\to\infty} \Exp{\Delta\rv{H}_n} = 0.\end{align*}  
Thus, $\Exp{\Delta \rv{H}_{\infty}} = 0$. 
By Markov's inequality, $\Probi{\Delta\rv{H}_{\infty} \geq \delta} \leq \Exp{\Delta\rv{H}_{\infty}}/\delta = 0$ for any $\delta >0$; consequently, $\Probi{\Delta\rv{H}_{\infty} = 0} = \Probi{\rv{H}_{\infty} = \hat{\rv{H}}_{\infty}} = 1$. Put another way,  $\hat{\rv{H}}_{\infty} = \rv{H}_{\infty}$ almost surely.

Recall that  $\rv{H}_{\infty}$ is a $\{0,1\}$ random variable with $\Probi{\rv{H}_{\infty} = 1} =  \ENTi{\rv{X}|\rv{Y}}$. 
Since $\hat{\rv{H}}_{\infty} = \rv{H}_{\infty}$ almost surely, and
\begin{align*}
	\Prob{\hat{\rv{H}}_n > 1-\epsilon} &= \frac{1}{N}\left|\left\{i: \ENT{\rv{U}_i| \rv{Q}_i, \rv{S}_0, \rv{S}_N} > 1-\epsilon \right\}\right|,\\
		\Prob{\hat{\rv{H}}_n < \epsilon} &= \frac{1}{N}\left|\left\{i: \ENT{\rv{U}_i| \rv{Q}_i, \rv{S}_0, \rv{S}_N} <\epsilon \right\}\right|,
\end{align*}
the proof is complete.
\end{IEEEproof}

\begin{IEEEproof}[Proof of \Cref{lem_DeltaH}]
By~\eqref{eq_conditioning reduces}, $\Delta\rv{H}_n \geq 0$, so $\Exp{\Delta \rv{H}_n} \geq 0$ as well.

Using the chain rule for conditional entropies and since the transformation $\rv{U}_1^N = \rv{X}_1^N G_N$ is one-to-one and onto,
	\[ 	
	\Exp{\rv{H}_n} = \frac{1}{N} \sum_{i=1}^N \ENT{\rv{U}_i|\rv{Q}_i} 
				   = \frac{\ENT{\rv{U}_1^N|\rv{Y}_1^N}}{N}  
				   = \frac{\ENT{\rv{X}_1^N|\rv{Y}_1^N}}{N}.
	\]
	Similarly, $\Exp{\hat{\rv{H}}_n} = \ENT{\rv{X}_1^N|\rv{Y}_1^N, \rv{S}_0,\rv{S}_N}/N$. 
	Thus, 
	\begin{align*}
	\Exp{\Delta \rv{H}_n} &= \frac{1}{N} \Big(\ENT{\rv{X}_1^N|\rv{Y}_1^N} - \ENT{\rv{X}_1^N|\rv{Y}_1^N, \rv{S}_0,\rv{S}_N}\Big) \\
						  &\eqann{a} \frac{1}{N} \Big(\ENT{\rv{S}_0,\rv{S}_N|\rv{Y}_1^N} - \ENT{\rv{S}_0,\rv{S}_N|\rv{X}_1^N,\rv{Y}_1^N}\Big) \\
						  &\eqann[\leq]{b} \frac{2 \log_2(|\mathcal{S}|)}{N}.
	\end{align*}
To see \eqannref{a}, note that for any $3$ random variables $\rv{A},\rv{B},\rv{C}$ we have $\ENT{\rv{A},\rv{B}|\rv{C}} = \ENT{\rv{A}|\rv{C}} + \ENT{\rv{B}|\rv{A},\rv{C}} = \ENT{\rv{B}|\rv{C}} + \ENT{\rv{A}|\rv{B},\rv{C}}$. Rearranging and setting $\rv{A} = \rv{X}_1^N$, $\rv{B} = (\rv{S}_0, \rv{S}_N)$ and $\rv{C} = \rv{Y}_1^N$ yields  \eqannref{a}. 
Inequality \eqannref{b} is since $\rv{S}_0,\rv{S}_N$ take values in the finite alphabet $\mathcal{S}$ and the conditional entropy is non-negative. 	

Combining these inequalities, and recalling that $N=2^n$, we obtain 
\[0 \leq \Exp{\Delta \rv{H}_n} \leq 2 \log_2(|\mathcal{S}|)/2^n.\]
This holds for any $n$. We take limits and use the sandwich rule to yield $\lim_{n\to\infty} \Exp{\Delta\rv{H}_n} = 0$,         as  desired.
\end{IEEEproof}

The following corollary is a direct consequence of the definition of almost-sure convergence, \Cref{lem_TV distance bounds}, \Cref{cor_polarization for FAIM from ST2016}, and \Cref{th_BSI entropy polarizes}. 
\begin{corollary}\label{cor_Zn Kn converge}\leavevmode
\begin{enumerate}
	\item The sequences $\rv{Z}_n$ and $\hat{\rv{Z}}_n$ polarize 
	to random variables $\rv{Z}_{\infty}$ and $\hat{\rv{Z}}_{\infty}$, respectively. Moreover, $\rv{Z}_{\infty}=\hat{\rv{Z}}_{\infty}=\rv{H}_{\infty}$ almost surely. 
	\item The sequences $\rv{K}_n$ and $\hat{\rv{K}}_n$ polarize  
	to random variables $\rv{K}_{\infty}$ and $\hat{\rv{K}}_{\infty}$, respectively. Moreover, $\rv{K}_{\infty}=\hat{\rv{K}}_{\infty}=1-\rv{H}_{\infty}$ almost surely. 
\end{enumerate}
\end{corollary}
\begin{IEEEproof}
The proofs of both items are essentially the same, so we prove only the first item. 

Recall the definition of almost-sure convergence of a sequence of random variables. Let $(\Omega,\mathcal{F},\mathbb{P})$ be a probability space, and let $\rv{A},\rv{A}_1,\rv{A}_2,\ldots$ be a sequence of $\mathcal{F}$-measurable random variables defined on this space. A random variable is a deterministic function from $\Omega$ to $\mathbb{R}$. We say that $\rv{A}_n$ converges to $\rv{A}$ almost surely if the set
\[ A = \left\{ \omega \in \Omega: \; \lim_{n\to\infty}\rv{A}_n(\omega) = \rv{A}(\omega)\right\}\]  
satisfies $\Prob{A} = 1$. 

Now, let $(\Omega,\mathcal{F},\mathbb{P})$ be the probability space in which $\rv{H}_n, \hat{\rv{H}}_n, \rv{Z}_n, \hat{\rv{Z}}_n$, $n =1,2,\ldots$ as well as $\rv{H}_{\infty}$ and $\hat{\rv{H}}_{\infty}$  
are defined. 

By \Cref{cor_polarization for FAIM from ST2016} and \Cref{th_BSI entropy polarizes}, $\rv{H}_n$ and $\hat{\rv{H}}_n$ converge almost surely to $\rv{H}_{\infty}$ and $\hat{\rv{H}}_{\infty}$, respectively, and $\rv{H}_{\infty} = \hat{\rv{H}}_{\infty}$ almost surely. Thus, we denote 
\[ H = \left\{ \omega\in\Omega: \; \lim_{n\to\infty}\rv{H}_n(\omega) = \lim_{n\to\infty}\hat{\rv{H}}_n(\omega) = \rv{H}_{\infty}(\omega)\right\}.\]
By definition of almost sure convergence, $\Probi{H} = 1$.    

Since $\rv{H}_{\infty}(\omega) \in \{0,1\}$ almost surely, we split $H = H_0 \cupdot H_1 \cupdot
H_{\emptyset}$, such that $\rv{H}_{\infty}(\omega) = 0$ for any $\omega \in H_0$;
$\rv{H}_{\infty}(\omega) = 1$ for any $\omega \in H_1$; and $H_{\emptyset}$ is a set of measure zero.  
By \Cref{lem_TV distance bounds}, we have $\rv{H}_n(\omega) \leq \rv{Z}_n(\omega) \leq
\sqrt{\rv{H}_n(\omega)}$ for any $\omega$. Thus, $\lim_{n\to\infty} \rv{Z}_n(\omega) = 0$ for all
$\omega \in H_0$ and $\lim_{n\to\infty} \rv{Z}_n(\omega) = 1$ for all $\omega \in H_1$. We conclude that $\rv{Z}_n$ converges almost surely to a $\{0,1\}$-random variable $\rv{Z}_{\infty}$ and $\rv{Z}_{\infty} = \rv{H}_{\infty}$ almost surely.  
Using similar arguments, $\hat{\rv{Z}}_n$ converges almost surely to a random variable $\hat{\rv{Z}}_{\infty}$ and $\hat{\rv{Z}}_{\infty} = \hat{\rv{H}}_{\infty}$ almost surely. By \Cref{th_BSI entropy polarizes}, $\hat{\rv{H}}_{\infty} = \rv{H}_{\infty}$ almost surely. 
\end{IEEEproof}

\subsection{Fast Polarization of the Bhattacharyya Process to $1$} 
In this section, we prove that the Bhattacharyya process $\rv{Z}_n$ of a FAIM process polarizes fast to $1$. 

\Cref{thm_fast polarization of Z to 1}, the main theorem of this section, relies on an inequality akin to~\eqref{eq_polarization bounds for TV} for the BSI total variation distance. We state the inequality in \Cref{prop_K inequalities for FAIM processes}, and postpone its proof to the end of the section. 
\begin{proposition}\label{prop_K inequalities for FAIM processes}
	Let $(\rv{X}_j,\rv{Y}_j,\rv{S}_j)$, $j \in \mathbb{Z}$ be a FAIM process. Then, 
	\begin{equation}
		\hat{\rv{K}}_{n+1} \leq \begin{cases} \psi(0) \hat{\rv{K}}_n^2 & \text{if } \rv{B}_{n+1} = 0 \\ 
 									  2 \hat{\rv{K}}_n  & \text{if } \rv{B}_{n+1} = 1. 
 									 \end{cases}	\label{eq_polarization bounds for TVhat cond}
	\end{equation} 
\end{proposition}
Here, $\psi(0)$ is as defined in~\eqref{eq_def of psi(n)}, i.e., 
\begin{equation} \psi(0) = \max_{\istate} \frac{1}{\pi_0(\istate)} = \max_{\mstate} \frac{1}{\pi_N(\mstate)} \geq 1. \label{eq_psi0} \end{equation}
Since the state sequence is stationary, finite-state, aperiodic, and irreducible, $\psi(0) < \infty$.

\begin{theorem} \label{thm_fast polarization of Z to 1}
	Let $(\rv{X}_j,\rv{Y}_j,\rv{S}_j)$, $j \in \mathbb{Z}$ be a FAIM process. Then $\rv{Z}_n$ polarizes fast to $1$ and for any $\beta < 1/2$, 
	\begin{equation} \lim_{N\to \infty} \frac{1}{N} \left|\left\{i: \BP{\rv{U}_i| \rv{Q}_i	} > 1-2^{-N^{\beta}} \right\} \right| = \ENTi{\rv{X}|\rv{Y}}. \label{eq_fast polarization of Zn to 1} \end{equation}
\end{theorem}
\begin{IEEEproof}
	Fix $\beta < 1/2$. 	By \Cref{cor_Zn Kn converge} and~\eqref{eq_polarization bounds for TVhat cond}, we can invoke \Cref{lem_simple proof} for $\hat{\rv{K}}_n$ with $E=1/2$. Consequently, $\hat{\rv{K}}_n$ polarizes fast to $0$, i.e., 
	\begin{align*} \lim_{n\to\infty} \Prob{\hat{\rv{K}}_n < 2^{-N^{\beta}}} &= \Prob{	\hat{\rv{K}}_{\infty}= 0} \\ &= \Prob{\rv{H}_{\infty} = 1} = \ENTi{\rv{X}|\rv{Y}}.\end{align*}
    For any $n$, by~\eqref{eq_TV relation 1},~\eqref{eq_bounds on BP}, 
    and~\eqref{eq_upper bound on KUQ},  
    \[ 1 - \rv{Z}_n \leq 1- \rv{H}_n \leq  \rv{K}_n \leq \hat{\rv{K}}_n.\] 
	Thus, 
	\[ \Prob{ \rv{Z}_n > 1-2^{-N^{\beta}}} \geq \Prob{\hat{\rv{K}}_n < 2^{-N^{\beta}}}.\] Taking limits, we obtain that 
	\[\liminf_{n\to\infty}\Probi{ \rv{Z}_n > 1-2^{-N^{\beta}}} \geq \ENTi{\rv{X}|\rv{Y}}.\] 
	
	On the other hand, by \Cref{cor_polarization for FAIM from ST2016}, 
	\[\lim_{n\to\infty}\Probi{ \rv{Z}_n <2^{-N^{\beta}}} = 1-\ENTi{\rv{X}|\rv{Y}}.\] Recalling that $\Probi{
    \rv{Z}_n <2^{-N^{\beta}}} + \Probi{ \rv{Z}_n > 1-2^{-N^{\beta}}} \leq 1$ for any $n$, we take
    limits to obtain $\limsup_{n\to\infty}\Probi{\rv{Z}_n > 1-2^{-N^{\beta}}} \leq
    \ENTi{\rv{X}|\rv{Y}}$. Therefore, we conclude that
	\[ \lim_{n\to\infty}\Prob{ \rv{Z}_n > 1-2^{-N^{\beta}}} = \ENTi{\rv{X}|\rv{Y}}.\]
	
	To obtain~\eqref{eq_fast polarization of Zn to 1}, note that by definition of the Bhattacharyya process, 
	\[ \Prob{\rv{Z}_n > 1-2^{-N^{\beta}}} = \frac{1}{N} \left|\left\{i: \BP{\rv{U}_i| \rv{Q}_i	} > 1-2^{-N^{\beta}} \right\} \right|.\] 
		Taking limits completes the proof.
\end{IEEEproof}

\begin{IEEEproof}[Proof of \Cref{prop_K inequalities for FAIM processes}] The proof follows along the lines of the proof of \Cref{prop_K is a supermartingale}.

Consider two adjacent blocks of length $N=2^n$ and let $i-1 = (\rv{B}_{1}\rv{B}_{2}\cdots\rv{B}_n)_2$. This is illustrated in \Cref{fig_two blocks}. 
	Recall from~\eqref{eq_P(UVQR) from P(XY)} that there is a function $f$ that depends on $i$ such that  
$(\rv{U}_i,\rv{Q}_i) = f(\rv{X}_1^N, \rv{Y}_1^N)$ and $(\rv{V}_i,\rv{R}_i) = f(\rv{X}_{N+1}^{2N}, \rv{Y}_{N+1}^{2N})$.
	By stationarity, 
	\begin{equation}
	\hat{\rv{K}}_n = \sol{\sum_{\istate,\mstate\in \mathcal{S}}} \pi_{N,0}(\mstate,\istate)\TVSn{\istate}{\mstate}{\rv{U}_i|\rv{Q}_i}	 = \sol{\sum_{\mstate,\fstate\in \mathcal{S}}} \pi_{2N,N}(\fstate,\mstate)\TVSn{\mstate}{\fstate}{\rv{V}_i|\rv{R}_i}.	\label{eq_Knhat with stationarity}
	\end{equation}

As in~\eqref{eq_def of Pabuq}, we denote
\[ P_{\istate}^{\fstate}(u,q) = \prrv{\rv{U}_i,\rv{Q}_i|\rv{S}_N,\rv{S}_0}{u,q|\fstate,\istate} = \prrv{\rv{V}_i,\rv{R}_i|\rv{S}_{2N},\rv{S}_N}{u,q|\fstate,\istate}.\]
The right-most equality is due to stationarity. 
We further denote $P_{\istate}^{\fstate}(s) = P_{\istate}^{\fstate}(0,s)+P_{\istate}^{\fstate}(1,s)$; in particular, $\sum_s P_{\istate}^{\fstate}(s) = 1$. 

Denote
\begin{align*}
\mu(\mstate) &=\pi_{2N|N}(\fstate|\mstate)\pi_{N|0}(\mstate|\istate)  \pi_0(\istate) \\
&= \frac{\pi_{2N,N}(\fstate,\mstate)  \cdot \pi_{N,0}(\mstate,\istate)}{\pi_{N}(\mstate)}.
\end{align*}
We deliberately omitted the dependence on $\istate,\fstate$ from this notation to simplify the expressions that follow. Observe that by~\eqref{eq_psi0}, 
\begin{equation} \mu(\mstate) \leq \psi(0)\cdot \pi_{2N,N}(\fstate,\mstate)  \cdot \pi_{N,0}(\mstate,\istate). \label{eq_upper bound on mub}\end{equation}
Also, since $\pi_N(\mstate)=\sum_{\istate\in\mathcal{S}} \pi_{N,0}(\mstate,\istate) = \sum_{\fstate\in\mathcal{S}} \pi_{2N,N}(\fstate,\mstate)$, we have
\begin{equation}
\sol{\sum_{\istate\in\mathcal{S}}} \mu(\mstate) = \pi_{2N,N}(\fstate,\mstate), 	 \quad
\sol{\sum_{\fstate\in\mathcal{S}}} \mu(\mstate) = \pi_{N,0}(\mstate,\istate). \label{eq_sums of mub}
\end{equation}

By~\eqref{eq_recursive computation of joint prob} and~\eqref{eq_def of Pabuq}, 
\begin{align}
&\pi_{2N,0}(\fstate,\istate) \prrv{\rv{U}_i,\rv{V}_i,\rv{Q}_i,\rv{R}_i|\rv{S}_{2N},\rv{S}_0}{u,v,q,r|\fstate,\istate}  \nonumber \\
&\quad= \pi_0(\istate) \pi_{2N|0}(\fstate|\istate) \prrv{\rv{U}_i,\rv{V}_i,\rv{Q}_i,\rv{R}_i|\rv{S}_{2N},\rv{S}_0}{u,v,q,r|\fstate,\istate}  \nonumber \\
 &\quad= \pi_0(\istate)\prrv{\rv{U}_i,\rv{V}_i,\rv{Q}_i,\rv{R}_i,\rv{S}_{2N}|\rv{S}_0}{u,v,q,r,\fstate|\istate}  \nonumber \\
 &\quad=  \pi_0(\istate)\sum_{\mstate \in \mathcal{S}} \prrv{\rv{U}_i,\rv{Q}_i,\rv{S}_{N}|\rv{S}_0}{u,q,\mstate |\istate} \prrv{\rv{V}_i,\rv{R}_i,\rv{S}_{2N}|\rv{S}_N}{v,r,\fstate|\mstate }  \nonumber \\
 &\quad=  \pi_0(\istate)\sum_{\mstate \in \mathcal{S}} \pi_{N|0}(\mstate|\istate)P_{\istate}^{\mstate }(u,q) \pi_{2N|N}(\fstate|\mstate ) P_{\mstate }^{\fstate}(v,r)  \nonumber \\
 &\quad= \sum_{\mstate \in \mathcal{S}} \mu(\mstate) P_{\istate}^{\mstate}(u,q)P_{\mstate}^{\fstate}(v,r). 
 \label{eq_PUVQR and mu}
\end{align}

 Set $\rv{T}_i = \rv{U}_i+\rv{V}_i$.  
 Using \eqref{eq_single step polarization for K}, a single-step polarization from $\hat{\rv{K}}_n$ to $\hat{\rv{K}}_{n+1}$ becomes
\[ \hat{\rv{K}}_{n+1} = \begin{dcases} \sol[r]{\sum_{\istate, \fstate \in \mathcal{S}}} \pi_{2N,0}(\fstate,\istate) \TVSn{\istate}{\fstate}{\rv{T}_i|\rv{Q}_i,\rv{R}_i} & \text{if } \rv{B}_{n+1} = 0 \\
 	\sol[r]{\sum_{\istate, \fstate \in \mathcal{S}}} \pi_{2N,0}(\fstate,\istate) \TVSn{\istate}{\fstate}{\rv{V}_i|\rv{T}_i,\rv{Q}_i,\rv{R}_i} & \text{if } \rv{B}_{n+1} = 1. 
 \end{dcases}\]
Here, $\TVSn{\istate}{\fstate}{\rv{T}_i|\rv{Q}_i,\rv{R}_i}$ and
$\TVSn{\istate}{\fstate}{\rv{V}_i|\rv{T}_i,\rv{Q}_i,\rv{R}_i}$ are computed as in~\eqref{eq_def of
Kab}, only  for a block of length $2N$ with initial state $\rv{S}_0 = \istate$ and final state
$\rv{S}_{2N} = \fstate$. At the middle of the block we have state $\rv{S}_N = \mstate$.
Using~\eqref{eq_P(STQR) from P(UVQR)}, we denote
\begin{align} \bar{P}_{\istate}^{\fstate}(t,v,q,r) &= \prrv{\rv{T}_i,\rv{V}_i,\rv{Q}_i,\rv{R}_i|\rv{S}_{2N},\rv{S}_0}{t,v,q,r|\fstate,\istate} \label{eq_Pac1} \\ &= \prrv{\rv{U}_i,\rv{V}_i,\rv{Q}_i,\rv{R}_i|\rv{S}_{2N},\rv{S}_0}{t+v,v,q,r|\fstate,\istate} \nonumber \\
\intertext{and}
\bar{P}_{\istate}^{\fstate}(t,q,r) &= \prrv{\rv{T}_i,\rv{Q}_i,\rv{R}_i|\rv{S}_{2N},\rv{S}_0}{t,q,r|\fstate,\istate} \label{eq_Pac2} \\
	 &= \sum_{v=0}^1  \bar{P}_{\istate}^{\fstate}(t,v,q,r). \nonumber
\end{align}                                                                             

Consider first the case $\rv{B}_{n+1} = 0$: 
\begin{align*}
&\pi_{2N,0}(\fstate,\istate)\TVSn{\istate}{\fstate}{\rv{T}_i|\rv{Q}_i,\rv{R}_i} \\ 
 &= \pi_{2N,0}(\fstate,\istate)\sol{\sum_{q,r}} \left| \bar{P}_{\istate}^{\fstate}(0,q,r) - \bar{P}_{\istate}^{\fstate}(1,q,r) \right|\\
 &= \sol{\sum_{q,r}} \left| \pi_{2N,0}(\fstate,\istate)\bar{P}_{\istate}^{\fstate}(0,q,r) - \pi_{2N,0}(\fstate,\istate)\bar{P}_{\istate}^{\fstate}(1,q,r) \right|\\
 &\eqann{a} \sol{\sum_{q,r}} \left| \sum_{\mstate \in \mathcal{S}}\mu(\mstate) \sum_{v=0}^1 P_{\mstate}^{\fstate}(v,r)(P_{\istate}^{\mstate}(v,q)-P_{\istate}^{\mstate}(v+1,q)) \right|\\
  &\eqann[\leq]{b} \sol{\sum_{\substack{q,r,\\ \mstate \in \mathcal{S}}}}\mu(\mstate) \left|  \sum_{v=0}^1 P_{\mstate}^{\fstate}(v,r)(P_{\istate}^{\mstate}(v,q)-P_{\istate}^{\mstate}(v+1,q)) \right|\\
  &= \sol{\sum_{\substack{q,r,\\ \mstate \in \mathcal{S}}}}\mu(\mstate) \Big|P_{\istate}^{\mstate}(0,q)-P_{\istate}^{\mstate}(1,q) \Big| \cdot \Big|P_{\mstate}^{\fstate}(0,r) - P_{\mstate}^{\fstate}(1,r)\Big|\\
  &= \sum_{\mstate \in \mathcal{S}}\mu(\mstate)  \TVSn{\istate}{\mstate}{\rv{U}_i|\rv{Q}_i} \TVSn{\mstate}{\fstate}{\rv{V}_i|\rv{R}_i}\\
  &\eqann[\leq]{c}  \psi(0) \sum_{\mstate \in \mathcal{S}} \Big(\pi_{2N,N}(\fstate,\mstate) \TVSn{\mstate}{\fstate}{\rv{V}_i|\rv{R}_i}\Big)  \cdot \Big(\pi_{N,0}(\mstate,\istate) \TVSn{\istate}{\mstate}{\rv{U}_i|\rv{Q}_i}\Big) \\
    &\eqann[\leq]{d}  \psi(0) \sum_{\mstate \in \mathcal{S}} \pi_{2N,N}(\fstate,\mstate) \TVSn{\mstate}{\fstate}{\rv{V}_i|\rv{R}_i}  \sol{\sum_{\mstate' \in \mathcal{S}}} \pi_{N,0}(\mstate',\istate) \TVSn{\istate}{\mstate'}{\rv{U}_i|\rv{Q}_i}, 
\end{align*}
where \eqannref{a} first expands $\bar{P}_{\istate}^{\fstate}(0,q,r)$ and $\bar{P}_{\istate}^{\fstate}(1,q,r)$  according to \eqref{eq_Pac2} and then~\eqref{eq_Pac1}, and finally applies~\eqref{eq_PUVQR and mu};  \eqannref{b} is by the triangle inequality; \eqannref{c} is by~\eqref{eq_upper bound on mub}; and \eqannref{d} is by the inequality $\sum_j a_j b_j \leq \sum_j a_j \sum_{j'}b_{j'}$, which holds for $a_j,b_j \geq 0$. 
By~\eqref{eq_Knhat with stationarity}, the sum over $\istate, \fstate \in \mathcal{S}$ yields
\[
\sol{\sum_{\istate, \fstate \in \mathcal{S}}} \pi_{2N,0}(\fstate,\istate) \TVSn{\istate}{\fstate}{\rv{T}_i|\rv{Q}_i,\rv{R}_i} \leq \psi(0) \hat{\rv{K}}_n^2. 
\]

Next, let $\rv{B}_{n+1} =1$. We have 
\begin{align*}
& \pi_{2N,0}(\fstate,\istate)\TVSn{\istate}{\fstate}{\rv{V}_i|\rv{T}_i,\rv{Q}_i,\rv{R}_i} \\
&= \pi_{2N,0}(\fstate,\istate)\sol{\sum_{t,q,r}} \left| \bar{P}_{\istate}^{\fstate}(t,0,q,r) - \bar{P}_{\istate}^{\fstate}(t,1,q,r) \right|\\
&= \sol{\sum_{t,q,r}} \left| \pi_{2N,0}(\fstate,\istate)\bar{P}_{\istate}^{\fstate}(t,0,q,r) - \pi_{2N,0}(\fstate,\istate)\bar{P}_{\istate}^{\fstate}(t,1,q,r) \right|\\
	&\eqann{a} \sol{\sum_{t,q,r}} \left|\sol{\sum_{\mstate \in \mathcal{S}}}\mu(\mstate)(P_{\istate}^{\mstate}(t,q)P_{\mstate}^{\fstate}(0,r) - P_{s}^{\mstate}(t+1,q)P_{\mstate}^{\fstate}(1,r)) \right|\\
	&\eqann{b} \frac{1}{2} \sol{\sum_{t,q,r}} \bigg|\sol{\sum_{\mstate \in \mathcal{S}}} \mu(\mstate) P_{\istate}^{\mstate}(q)(P_{\mstate}^{\fstate}(0,r) - P_{\mstate}^{\fstate}(1,r)) \\
	&\qquad \qquad+ \sol{\sum_{\mstate \in \mathcal{S}}} \mu(\mstate) P_{\mstate}^{\fstate}(r)(P_{\istate}^{\mstate}(t,q) - P_{\istate}^{\mstate}(t+1,q)) \bigg|\\
	&\eqann[\leq]{c}  \sol{\sum_{\substack{q,\\ \mstate \in \mathcal{S}}}} \mu(\mstate) P_{\istate}^{\mstate}(q) \left(\sol{\sum_{r}} \bigg| P_{\mstate}^{\fstate}(0,r) - P_{\mstate}^{\fstate}(1,r)\bigg| \right) \\
	&\quad +  \sol{\sum_{\substack{r,\\ \mstate \in \mathcal{S}}}}\mu(\mstate) P_{\mstate}^{\fstate}(r) \left(\sol{\sum_{q }} \bigg| P_{\istate}^{\mstate}(0,q) - P_{\istate}^{\mstate}(1,q)\bigg|\right)  \\
	&= \sum_{\mstate\in\mathcal{S}} \mu(\mstate) \TVSn{\mstate}{\fstate}{\rv{V}_i|\rv{R}_i}  + \sum_{\mstate\in\mathcal{S}} \mu(\mstate) \TVSn{\istate}{\mstate}{\rv{U}_i|\rv{Q}_i},
\end{align*}
where \eqannref{a} first expands $\bar{P}_{\istate}^{\fstate}(t,0,q,r)$ and $\bar{P}_{\istate}^{\fstate}(t,1,q,r)$ according to~\eqref{eq_Pac1}, and then applies~\eqref{eq_PUVQR and mu};
 \eqannref{b} is by~\eqref{eq_abcd equality}; and
\eqannref{c} is by the triangle inequality. Since $\mu(\mstate)$ depends on $\istate, \fstate$, we
use~\eqref{eq_sums of mub} to obtain
\begin{align*} \sol{\sum_{\istate,\mstate,\fstate \in \mathcal{S}}} \mu(\mstate)\TVSn{\mstate}{\fstate}{\rv{V}_i|\rv{R}_i} &= \sol{\sum_{\mstate,\fstate \in \mathcal{S}}} \pi_{2N,N}(\fstate,\mstate)\TVSn{\mstate}{\fstate}{\rv{V}_i|\rv{R}_i} = \hat{\rv{K}}_n, \\
\sol{\sum_{\istate,\mstate,\fstate \in \mathcal{S}}} \mu(\mstate)\TVSn{\istate}{\mstate}{\rv{U}_i|\rv{Q}_i} &= \sol{\sum_{\istate,\mstate \in \mathcal{S}}} \pi_{N,0}(\mstate,\istate)\TVSn{\istate}{\mstate}{\rv{U}_i|\rv{Q}_i} = \hat{\rv{K}}_n. \end{align*}
Thus, 
\[
\sol{\sum_{\istate, \fstate \in \mathcal{S}}} \pi_{2N,0}(\fstate,\istate) \TVSn{\istate}{\fstate}{\rv{V}_i|\rv{T}_i,\rv{Q}_i,\rv{R}_i} \leq 2 \hat{\rv{K}}_n. 
\] 
This completes the proof. 
\end{IEEEproof}

\subsection{Fast Polarization of the BSI Bhattacharyya Process}
Fast polarization of the Bhattacharyya process was established in \Cref{cor_polarization for FAIM from ST2016} and \Cref{thm_fast polarization of Z to 1}. Implicitly, however, we have also obtained fast polarization of the BSI-Bhattacharyya process $\hat{\rv{Z}}_n$, both to $0$ and $1$. We now make this explicit. 

\begin{corollary}
	Let $(\rv{X}_j,\rv{Y}_j,\rv{S}_j)$, $j \in \mathbb{Z}$ be a FAIM process. Then $\hat{\rv{Z}}_n$ polarizes fast both to $0$ and to $1$ with any $\beta < 1/2$.
\end{corollary}

\begin{IEEEproof}  Polarization of $\hat{\rv{Z}}_n$ was obtained directly in \Cref{cor_Zn Kn converge}. By~\eqref{eq_upper bound on KUQ}, $\rv{Z}_n \geq \hat{\rv{Z}}_n$. Since $\rv{Z}_n$ polarizes fast to $0$ with any $\beta < 1/2$, so must $\hat{\rv{Z}}_n$. We obtain fast polarization of $\hat{\rv{Z}}_n$ to $1$ by replacing the Bhattacharyya parameter with its BSI counterpart in the proof of \Cref{thm_fast polarization of Z to 1}. 
\end{IEEEproof}

\appendices
\section{Auxiliary Proofs for \Cref{sec_polar toolbox}}\label{app_proof of TV lemma}
For $\theta \in [0,1/2]$ we denote
		\begin{align*}
			k(\theta) &= |\theta - (1-\theta)| = 1-2\theta, \\ 
			h(\theta) &= -\theta \log_2 \theta - (1-\theta)\log_2(1-\theta),\\
			z(\theta) &= 2\sqrt{\theta(1-\theta)}. 
		\end{align*}
We will need the following lemmas. 
\begin{lemma}\label{lem_h2 k2 leq 1}
    For $\theta \in [0, 1/2]$, we have $ z^2(\theta) \leq h(\theta) \leq z(\theta)$. 
\end{lemma}
\begin{IEEEproof}
    We plot $z^2(\theta)$, $h(\theta)$ and $z(\theta)$ in \Cref{fig_hx leq gx}; indeed $z^2(\theta)
    \leq h(\theta) \leq z(\theta)$ for $0 \leq \theta \leq 1/2$. We now prove this formally.

   The left-most inequality is obvious for $\theta=0$. Next, observe that $h(\theta)/\theta$ is
       convex-$\cup$ in $(0,1/2]$. To see this,  we turn to its second order derivative: 
       \[ \left(\frac{h(\theta)}{\theta}\right)'' = \frac{-(\theta+2(1-\theta)\ln(1-\theta))
       }{(1-\theta)\theta^3 \ln 2}.\]
       We claim that it is nonnegative for $\theta \in (0,1/2]$, which will imply that $h(\theta)/\theta$
       is indeed convex-$\cup$ in $(0,1/2]$. The denominator is nonnegative, so it
       remains to show that the numerator is nonnegative as well. Negating the numerator yields $\tau(\theta) = \theta
       + 2(1-\theta)\ln(1-\theta)$, which is convex-$\cup$ in $[0,1/2]$ as it is a positive sum of two
       convex-$\cup$ functions. Since $\tau(0) = 0$ and $\tau(1/2) = 1/2-\ln 2 < 0$, by Jensen's
       inequality,
       \begin{align*}
               \tau(\theta) &= \tau((1-2\theta) \cdot 0 + 2\theta \cdot 1/2) \\
                                & \leq  (1-2\theta) \cdot \tau(0)  + 2\theta\cdot \tau(1/2)\\
                                &< 0
           \end{align*}
           for any $\theta \in (0,1/2]$. This implies that the numerator of the second-order derivative is nonnegative,
           establishing convexity of $h(\theta)/\theta$.

           Consequently, $h(\theta)/\theta$ satisfies
           the gradient inequality (\cite[Theorem 7.6]{Beck_NonlinearOptimization}) by which
           \begin{align*} \frac{h(\theta)}{\theta} &\geq \frac{h(1/2)}{1/2} +
                   \left. \left( \frac{h(\theta)}{\theta} \right)'\right|_{\theta=1/2} (\theta-1/2) \\
                       &= 2-4(\theta-1/2) \\
                       &= 4(1-\theta).\end{align*}
                               This holds for any $\theta \in (0,1/2]$.  Rearranging yields $h(\theta) \geq
                                       4\theta(1-\theta) = z^2(\theta)$, which holds for any $\theta \in [0,1/2]$.

	For the right-most inequality, denote $g(\theta) = h(\theta)-z(\theta)$. Since $g(0) = g(1/2) = 0$, it suffices to show that $g(\theta)$ has a single stationary point in $(0,1/2)$, and that this point is a minimum. 
		
	The stationary points of $g(\theta)$ are the zeros of its derivative
	\[ g'(\theta) = \log_2\left(\frac{1-\theta}{\theta}\right)  - \frac{1-2\theta}{\sqrt{\theta(1-\theta)}}.\] 
	Recalling that $\theta \in [0,1/2]$, 
	\[ g'''(\theta) = \frac{(1-2\theta)(4\sqrt{\theta(1-\theta)}-\ln 8)}{4\left(\sqrt{\theta(1-\theta)}\right)^5\ln 2}\leq 0,\]
	since \[4\sqrt{\theta(1-\theta)}-\ln 8 < 4\sqrt{\theta(1-\theta)}-2 \leq 0.\] Hence, $g'(\theta)$ is concave-$\cap$ in $[0,1/2]$. 	Observe that $g'(1/2) = 0$ and $\lim_{\theta\to 0} g'(\theta) = -\infty$, so $g'(\theta)$ can assume the value $0$ for at most one point in $(0,1/2)$. Assume to the contrary that $g'(\theta)<0$ for all $\theta \in (0,1/2)$. Then, $g(\theta)$ has no stationary points in $(0,1/2)$, which, by the mean value theorem, contradicts $g(0) = g(1/2) = 0$. We conclude that $g'(\theta_0) = 0 $ for some $\theta_0 \in (0,1/2)$. Consequently, $\theta_0$ is a stationary point of $g(\theta)$. Since $g'(\theta)$ is concave-$\cap$ and $g'(1/2) = g'(\theta_0) = 0$, then $g'(\theta) > g'(\theta_0)$ for $\theta_0<\theta<1/2$ and  $g'(\theta)<g'(\theta_0)$ for $0<\theta<\theta_0$. This implies that $g(\theta) \geq g(\theta_0)$ for any $\theta \in [0,1/2]$; i.e., $\theta_0$ is the single minimum of $g(\theta)$ in $[0,1/2]$.  	
	\begin{figure}
	\begin{center}
%
%
\begin{tikzpicture}

\begin{axis}[%
width=8cm,
height=4cm,
scale only axis,
xmin=0,
xmax=0.5,
xlabel=$\theta$,
ymin=0,
ymax=1,
legend style={legend cell align=right,align=right,draw=white!15!black, legend pos = south east}
]

\addplot [color=red,dashed,thick]
  table[row sep=crcr]{%
0	0\\
0.01	0.198997487421324\\
0.02	0.28\\
0.03	0.34117444218464\\
0.04	0.391918358845308\\
0.05	0.435889894354067\\
0.06	0.474973683481517\\
0.07	0.510294032886923\\
0.08	0.542586398650021\\
0.09	0.572363520850167\\
0.1	0.6\\
0.11	0.625779513886481\\
0.12	0.649923072370877\\
0.13	0.672606868832009\\
0.14	0.693974062915899\\
0.15	0.714142842854285\\
0.16	0.733212111192934\\
0.17	0.751265598839718\\
0.18	0.768374908491942\\
0.19	0.784601809837321\\
0.2	0.8\\
0.21	0.814616474176652\\
0.22	0.828492607088319\\
0.23	0.841665016500033\\
0.24	0.854166260162505\\
0.25	0.866025403784439\\
0.26	0.877268487978452\\
0.27	0.887918915216925\\
0.28	0.897997772825746\\
0.29	0.907524104363074\\
0.3	0.916515138991168\\
0.31	0.924986486387774\\
0.32	0.932952303175248\\
0.33	0.9404254356407\\
0.34	0.947417542586161\\
0.35	0.953939201416946\\
0.36	0.96\\
0.37	0.965608616365865\\
0.38	0.970772887960928\\
0.39	0.975499871860576\\
0.4	0.979795897113271\\
0.41	0.983666610188635\\
0.42	0.987117014340245\\
0.43	0.990151503558925\\
0.44	0.992773891679269\\
0.45	0.99498743710662\\
0.46	0.996794863550169\\
0.47	0.998198377077422\\
0.48	0.999199679743744\\
0.49	0.999799979995999\\
0.5	1\\
};
\addlegendentry{$z(\theta)$};

\addplot [color=blue,solid,thick]
  table[row sep=crcr]{%
0	0\\
0.01	0.0807931358959112\\
0.02	0.141440542541821\\
0.03	0.194391857831576\\
0.04	0.242292189082415\\
0.05	0.286396957115956\\
0.06	0.327444919154476\\
0.07	0.365923650900223\\
0.08	0.402179190202273\\
0.09	0.436469817064103\\
0.1	0.468995593589281\\
0.11	0.499915958164528\\
0.12	0.529360865287364\\
0.13	0.557438185027989\\
0.14	0.584238811642856\\
0.15	0.6098403047164\\
0.16	0.634309554640566\\
0.17	0.65770477874422\\
0.18	0.68007704572828\\
0.19	0.701471459883897\\
0.2	0.721928094887362\\
0.21	0.741482739931274\\
0.22	0.760167502961966\\
0.23	0.778011303546538\\
0.24	0.795040279384522\\
0.25	0.811278124459133\\
0.26	0.826746372492618\\
0.27	0.841464636208176\\
0.28	0.855450810560131\\
0.29	0.868721246339405\\
0.3	0.881290899230693\\
0.31	0.893173458377857\\
0.32	0.904381457724494\\
0.33	0.914926372779727\\
0.34	0.92481870497303\\
0.35	0.934068055375491\\
0.36	0.942683189255492\\
0.37	0.950672092687066\\
0.38	0.9580420222263\\
0.39	0.964799548505087\\
0.4	0.970950594454669\\
0.41	0.976500468757824\\
0.42	0.981453895033654\\
0.43	0.98581503717892\\
0.44	0.989587521222056\\
0.45	0.992774453987808\\
0.46	0.995378438820226\\
0.47	0.99740158856774\\
0.48	0.998845535995202\\
0.49	0.99971144175281\\
0.5	1\\
};
\addlegendentry{$h(\theta)$};

 \addplot [color=green!50!black,dotted,thick]
  table[row sep=crcr]{%
 0     0                 \\ 
 0.01  0.039600000000000 \\ 
 0.02  0.078400000000000 \\ 
 0.03  0.116400000000000 \\ 
 0.04  0.153600000000000 \\ 
 0.05  0.190000000000000 \\ 
 0.06  0.225600000000000 \\ 
 0.07  0.260400000000000 \\ 
 0.08  0.294399999999999 \\ 
 0.09  0.327600000000000 \\ 
 0.1	  0.360000000000000 \\ 
 0.11  0.391600000000001 \\ 
 0.12  0.422400000000000 \\ 
 0.13  0.452399999999999 \\      
 0.14  0.481600000000000 \\ 
 0.15  0.510000000000000 \\ 
 0.16  0.537599999999999 \\ 
 0.17  0.564400000000000 \\ 
 0.18  0.590400000000000 \\ 
 0.19  0.615600000000000 \\ 
 0.2	  0.640000000000000 \\ 
 0.21  0.663600000000000 \\ 
 0.22  0.686400000000000 \\ 
 0.23  0.708400000000001 \\ 
 0.24  0.729600000000000 \\ 
 0.25  0.750000000000001 \\ 
 0.26  0.769599999999999 \\ 
 0.27  0.788400000000001 \\ 
 0.28  0.806400000000000 \\ 
 0.29  0.823600000000000 \\ 
 0.3	  0.840000000000000 \\ 
 0.31  0.855600000000000 \\ 
 0.32  0.870400000000000 \\ 
 0.33  0.884400000000000 \\ 
 0.34  0.897600000000000 \\ 
 0.35  0.910000000000001 \\ 
 0.36  0.921600000000000 \\ 
 0.37  0.932400000000000 \\ 
 0.38  0.942400000000000 \\ 
 0.39  0.951600000000000 \\ 
 0.4	  0.960000000000000 \\ 
 0.41  0.967600000000000 \\ 
 0.42  0.974399999999999 \\ 
 0.43  0.980400000000000 \\ 
 0.44  0.985600000000001 \\ 
 0.45  0.990000000000000 \\ 
 0.46  0.993600000000000 \\ 
 0.47  0.996399999999999 \\ 
 0.48  0.998400000000001 \\ 
 0.49  0.999600000000000 \\    
 0.5   1.000000000000000 \\ 
};
\addlegendentry{$z^2(\theta)$};

\end{axis}
\end{tikzpicture}%
	\end{center} 
    \caption{Illustration that $z^2(\theta) \leq h(\theta) \leq z(\theta)$ for $0 \leq \theta \leq 1/2$.}\label{fig_hx leq gx}
	\end{figure}
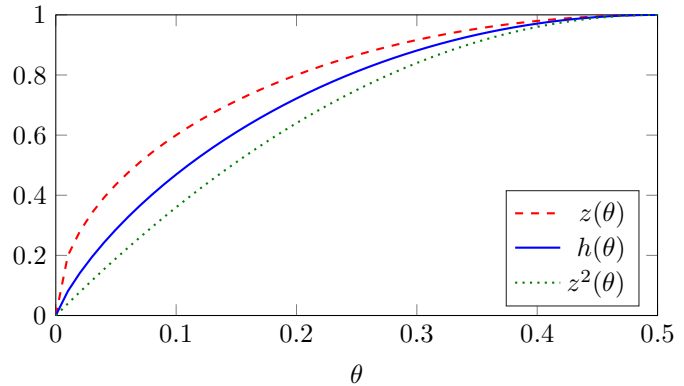  
\end{IEEEproof}
\begin{lemma} \label{lem_h+k leq}
    For $\theta \in [0,1/2]$, we have $k(\theta) +  h(\theta) \geq 1$.
\end{lemma}
\begin{proof}
    Both $k(\theta)$ and $h(\theta)$ are continuous and concave-$\cap$ functions in $[0,1/2]$. Therefore,
    $\eta(\theta) = k(\theta) + h(\theta)$ is also concave-$\cap$ in this region. Observe that
    $\eta(0) = \eta(1/2)=1$. Any $\theta \in [0,1/2]$ can be written as a convex combination of $0$
    and $1/2$, since $\theta = (1-2\theta) \cdot 0  + (2\theta)\cdot (1/2)$. Thus, by Jensen's
    inequality for concave-$\cap$ functions, $\eta(\theta) \geq (1-2\theta)\eta(0)
    + (2\eta)\eta(1/2) = 1$ for any $\theta\in[0,1/2]$.    
\end{proof}

\begin{IEEEproof}[Proof of \Cref{lem_TV distance bounds}]
		For any $q$, denote \[\theta = \theta(q) = \min\{\pr{U|Q}{0|q},\pr{U|Q}{1|q}\}.\] Accordingly,  $1-\theta = \max\{\pr{U|Q}{0|q},\pr{U|Q}{1|q}\}$ and $\theta \in [0,1/2]$.
		The various distribution parameters are expectations of functions of $\theta$:			
		\begin{align*} 
			\Pe{\rv{U}|\rv{Q}} &= \sum_q \pr{Q}{q} \theta, \\
			\TV{\rv{U}|\rv{Q}} &= \sum_q \pr{Q}{q}k(\theta),\\
			\ENT{\rv{U}|\rv{Q}} &= \sum_q \pr{Q}{q}h(\theta),  \\
			\BP{\rv{U}|\rv{Q}} &= \sum_q \pr{Q}{q}z(\theta).			
		\end{align*}

We directly obtain the equality in~\eqref{eq_TV relation 1}, as 
		\begin{align*} 
			\TV{\rv{U}|\rv{Q}} &= \sum_q \pr{Q}{q}(1-2\theta)\\ 
							   &= 1 - 2\sum_q \pr{Q}{q}\theta\\
							   & = 1 - 2\Pe{\rv{U}|\rv{Q}}.
		\end{align*}
        The inequality of~\eqref{eq_TV relation 1} is a consequence of \Cref{lem_h+k leq}, as
\begin{align*} 
    \TV{\rv{U}|\rv{Q}} + \ENT{\rv{U}|\rv{Q}} &= \sum_q \pr{Q}{q} (k(\theta) + h(\theta)) \\ 
                                             &\geq 1. 
\end{align*}

The right-most inequalities of~\eqref{eq_bounds on BP} and~\eqref{eq_bounds on TV upper} are  immediate consequences of \Cref{lem_h2 k2 leq 1}. Thus, we concentrate on the left-most inequalities.

For the left-most inequality of~\eqref{eq_bounds on BP}, we employ Jensen's inequality for the
convex-$\cup$ function $x \mapsto x^2$ and the
inequality $z^2(\theta) \leq h(\theta)$ from \Cref{lem_h2 k2 leq 1} to obtain 
\begin{align*}
    \BP{\rv{U}|\rv{Q}}^2 &= \left( \sum_q \pr{Q}{q} z(\theta) \right)^2 \\
                         &\leq \sum_q \pr{Q}{q} z^2(\theta) \\
                         &\leq \sum_q \pr{Q}{q} h(\theta) \\ 
                         &= \ENT{\rv{U}|\rv{Q}}. 
\end{align*}

	For the left-most inequality of~\eqref{eq_bounds on TV upper}, observe that 
	\begin{align*} z^2(\theta) + k^2(\theta) &=  4\theta(1-\theta) + (\theta-(1-\theta))^2\\
											   &= \theta^2 + 2\theta(1-\theta) + (1-\theta)^2 \\
											   &= (\theta + (1-\theta))^2\\
											   &= 1.	
	\end{align*}
	Using Jensen's inequality twice for the convex-$\cup$ function $x\mapsto x^2$, 
\[ \BP{\rv{U}|\rv{Q}}^2 + \TV{\rv{U}|\rv{Q}}^2 \leq \sum_q \pr{Q}{q} (z^2(\theta)	+ k^2(\theta)) = 1.\]
This implies the left-most inequality of~\eqref{eq_bounds on TV upper}. 
\end{IEEEproof}

\begin{IEEEproof}[Proof of \Cref{lem_effect of conditioning}] 
We obtain the joint distribution of $(\rv{U},\rv{Q})$ by marginalizing $\pr{U,Q,S}{}$, 
\[ \pr{U,Q}{u,q} = \sum_{s} \pr{U,Q,S}{u,q,s}.\]

The triangle inequality yields~\eqref{eq_TV conditioning}:
\begin{align*}
	\TV{\rv{U}|\rv{Q}} &= \sum_q |\pr{U,Q}{0,q}-\pr{U,Q}{1,q}| \\ 
					   &= \sum_q \left| \sum_s \big(\pr{U,Q,S}{0,q,s} - \pr{U,Q,S}{1,q,s} \big) \right| \\
					   &\leq \sum_{q,s}  \left|\pr{U,Q,S}{0,q,s} - \pr{U,Q,S}{1,q,s}\right|\\
					   &= \TV{\rv{U}|\rv{Q},\rv{S}}. 
\end{align*}

We derive~\eqref{eq_BP conditioning} using the Cauchy-Schwartz inequality: 
\begin{align*}
	\BP{\rv{U}|\rv{Q}} &= 2\sum_q \sqrt{\pr{U,Q}{0,q}\pr{U,Q}{1,q}} \\
					   &= 2\sum_q \sqrt{\sum_s \pr{U,Q,S}{0,q,s} \sum_{s'}\pr{U,Q,S}{1,q,s'}} \\
					   &\geq 2\sum_{q,s}\sqrt{\pr{U,Q,S}{0,q,s}}\sqrt{\pr{U,Q,S}{1,q,s}} \\
					   &=\BP{\rv{U}|\rv{Q},\rv{S}}.
\end{align*}

Inequality~\eqref{eq_ENT conditioning} is a consequence of Jensen's inequality for the concave-$\cap$ function $x\mapsto -x\log_2 x$. A proof can be found in~\cite[Theorem 2.6.5]{cover_thomas}. 
\end{IEEEproof}

\section{Extension to the Non-Binary Case}\label{app_non binary extension}
Our results are readily extended to the non-binary case. Here, $\rv{X}_j, j \in \mathbb{Z}$ take values in an alphabet $\mathcal{U}$ with $|\mathcal{U}| = L$. As in~\cite{sasoglu_nonbinary}, 
we use Ar\i{}kan's polarization transform in the non-binary case, replacing addition of $L$-ary numbers with modulo-$L$ addition. Thus~\eqref{eq_defs of UVQR} applies in the non-binary case; addition in~\eqref{eq_defs of UVQR U} and~\eqref{eq_defs of UVQR V} is modulo-$L$.  

First, we extend the distribution parameters from \Cref{sec_distribution parameters} to non-binary $\rv{U}$. We do this while keeping key properties that allows their use in polar code analysis. Then, we consider their fast polarization. Fast polarization of the Bhattacharyya process was established in~\cite[Chapter 3]{sasoglu_thesis}. We show that the total variation process satisfies the conditions for fast polarization required for \Cref{lem_simple proof}.

\subsection{Non-Binary Distribution Parameters}
The three distribution parameters we consider --- Bhattacharyya parameter, total variation distance, and conditional entropy --- were all defined for random variable pairs $(\rv{U},\rv{Q})$ where $\rv{U}$ is binary. We now show how to extend them to the case where $\rv{U}$ may take values in an arbitrary finite alphabet $\mathcal{U}$. We denote $|\mathcal{U}| = L$. 

There are two properties of the distribution parameters that are crucial for the analysis of polar codes. First, they are to take values in $[0,1]$. Second, when each of them approaches one of the extreme values, so should the others. The suggested non-binary extension satisfies these properties. The extension for the Bhattacharyya parameter and conditional entropy are based on~\cite{sasoglu_nonbinary}, which ensued the study of non-binary polar codes (see also~\cite[Chapter 3]{sasoglu_thesis}). 

Denote 
\begin{align*}
	\BPL{\rv{U}|\rv{Q}} &= \sum_q \sol{\sum_{u' \neq u}} \frac{\pr{Q}{q}}{L-1} \sqrt{\pr{U|Q}{u|q}\pr{U|Q}{u'|q}}, \\
	\TVL{\rv{U}|\rv{Q}} &= \sum_q \sol{\sum_{u' \neq u}} \frac{\pr{Q}{q}}{L-1} \frac{\left|\pr{U|Q}{u|q}-\pr{U|Q}{u'|q}\right|}{2}, \\ 
	\ENTL{\rv{U}|\rv{Q}} &= -\sum_q \sum_u \pr{Q}{q} \pr{U|Q}{u|q}\log_L \pr{U|Q}{u|q}. 
\end{align*}
As expected, when $L=2$ these coincide with~\eqref{eq_def of Z}--\eqref{eq_def of condent}. All three parameters are in $[0,1]$. This is well-known for the conditional entropy (see, e.g., \cite[Chapter 2]{cover_thomas}); for the total variation distance and the Bhattacharyya parameter, see the proof of \Cref{lem_non binary parameters}, below. The three parameters achieve their extreme values either when $\pr{U|Q}{u|q} =1/L$ for all $u$ or when there is some $u_0 \in \mathcal{U}$ such that $\pr{U|Q}{u_0|q} = 1$ and  $\pr{U|Q}{u|q}=0$ for $u\neq u_0$. 

The consequences of \Cref{lem_TV distance bounds} apply in the non-binary case as well. That is, when one of the three parameters approaches an extreme values, so do the other two. This is a consequence of the following lemma. 

\begin{lemma}\label{lem_non binary parameters}
The non-binary total variation distance, probability of error, conditional entropy, and Bhattacharyya parameter  are related by
\begin{subequations}\label{eq_L bounds}
\begin{align} \BPL{\rv{U}|\rv{Q}}^2 &\leq \ENTL{\rv{U}|\rv{Q}} \leq \log_L(1+(L-1)\BPL{\rv{U}|\rv{Q}}),\label{eq_BPL ENTL bounds} 	\\
1-\BPL{\rv{U}|\rv{Q}} &\leq \TVL{\rv{U}|\rv{Q}} \leq \sqrt{ 1 -  \BPL{\rv{U}|\rv{Q}}^2}.\label{eq_TVL BPL ENTL bounds} 		
\end{align}
\end{subequations}
\end{lemma}
\begin{remark}
Inequality~\eqref{eq_TVL BPL ENTL bounds} was also independently derived for the binary symmetric case in \cite{Dumer_stepped}, using a different proof. 
\end{remark}

\begin{IEEEproof}
The inequalities in~\eqref{eq_BPL ENTL bounds} were derived in~\cite[Proposition 3.3]{sasoglu_thesis}. Thus, we concentrate on showing~\eqref{eq_TVL BPL ENTL bounds}.

To see the right-most inequality of~\eqref{eq_TVL BPL ENTL bounds}, note that 
\[ \sum_q \sol{\sum_{u'\neq u}}\frac{\pr{Q}{q}}{L(L-1)} =  \sum_{q,u} \sol{\sum_{\underline{u'} \neq u}}  \frac{\pr{Q}{q}}{L(L-1)} = 1.\] 
Thus, by Jensen's inequality,
\begin{align*} \frac{\BPL{\rv{U}|\rv{Q}}^2}{L^2} &\leq \sum_{q,u} \sol{\sum_{\underline{u'} \neq u}}  \frac{\pr{Q}{q}}{L(L-1)} \left(\sqrt{\pr{U|Q}{u|q}\pr{U|Q}{u'|q}}\right)^2, \\	
	\frac{\TVL{\rv{U}|\rv{Q}}^2}{L^2} &\leq \sum_{q,u} \sol{\sum_{\underline{u'} \neq u}}  \frac{\pr{Q}{q}}{L(L-1)} \left(\frac{\pr{U|Q}{u|q}-\pr{U|Q}{u'|q}}{2}\right)^2.
\end{align*}
Next, observe that  
\begin{align*}
\left(\sqrt{\pr{U|Q}{u|q}\pr{U|Q}{u'|q}}\right)^2 &+ \left(\frac{\pr{U|Q}{u|q}-\pr{U|Q}{u'|q}}{2}\right)^2 \\ 
&= \left( \frac{\pr{U|Q}{u|q}+\pr{U|Q}{u'|q}}{2}\right)^2 
\end{align*}
and that subject to the constraint $\sum_u \pr{U|Q}{u|q} =1$, we have 
\[ \sum_u \sol{\sum_{\underline{u'}\neq u}} \left(\frac{\pr{U|Q}{u|q}+\pr{U|Q}{u'|q}}{2}\right)^2 \leq \frac{L(L-1)}{L^2}.\]
This can be seen using Lagrange multipliers; the maximum value is obtained with equality when $\pr{U|Q}{u|q} = 1/L$ for all $u \in \mathcal{U}$.  
Thus, we obtain 
\[	\BPL{\rv{U}|\rv{Q}}^2 + \TVL{\rv{U}|\rv{Q}}^2 \leq 1,\] 
which implies the right-most inequality of~\eqref{eq_TVL BPL ENTL bounds}. This also shows that indeed $\BPL{\rv{U}|\rv{Q}} \leq 1$ and $ \TVL{\rv{U}|\rv{Q}} \leq 1$. 

For the left-most inequality of~\eqref{eq_TVL BPL ENTL bounds}, observe that for any $a,b\geq 0$ we have $\sqrt{ab} \geq \min\{a,b\}$, by which
\begin{align*}
\frac{|a-b|}{2} + \sqrt{ab} &= \frac{\max\{a,b\} - \min\{a,b\}}{2} + \sqrt{ab} \\
				    &\geq \frac{\max\{a,b\} - \min\{a,b\} + 2 \min\{a,b\}}{2} \\
				    &= \frac{\max\{a,b\} + \min\{a,b\}}{2} \\
				    &= \frac{a+b}{2}.
\end{align*}
Since 
\[ \sol{\sum_{u' \neq u}} \pr{U|Q}{u|q} = \sol{\sum_{u' \neq u}} \pr{U|Q}{u'|q} = L-1,\] 
we have 
\begin{align*}
\BPL{\rv{U}|\rv{Q}} + \TVL{\rv{U}|\rv{Q}} 
&\geq \sum_q \pr{Q}{q} \sol{\sum_{u' \neq u}} \frac{\pr{U|Q}{u|q}+\pr{U|Q}{u'|q}}{2(L-1)}\\
&= 1. 
\end{align*}
This yields the left-most inequality of~\eqref{eq_TVL BPL ENTL bounds}.	
\end{IEEEproof}

Indeed, inequalities~\eqref{eq_L bounds} imply that when either $\BPL{\rv{U}|\rv{Q}}$ or $\ENTL{\rv{U}|\rv{Q}}$ approach $0$ or $1$ then $\TVL{\rv{U}|\rv{Q}}$ approaches $1$ or $0$, respectively, and vice versa. 

In the binary case, the total-variation distance and the probability of error were related by~\eqref{eq_TV relation 1}. In the non-binary case, the probability of error is given by 
\[ \PeL{\rv{U}|\rv{Q}} = \sum_q \pr{Q}{q}(1-\max_u{\pr{U|Q}{u|q}}).\] 
The non-binary probability of error and total variation distance are related, as shown in the following lemma. \begin{lemma}\label{lem_non binary TV and Pe}
	The non-binary probability of error and total variation distance are related by 
	\[ \TVL{\rv{U}|\rv{Q}} \leq 1 - \frac{2}{L-1}\PeL{\rv{U}|\rv{Q}}.\] 
\end{lemma}

\begin{IEEEproof}
	Let $\mathcal{U} = \{0,1, \ldots, L-1\}$. Without loss of generality we assume that, for a given $q \in \mathcal{Q}$, 
	\begin{equation} \pr{U|Q}{0|q} \leq \pr{U|Q}{1|q} \leq \cdots \leq \pr{U|Q}{L-1|q}.\label{eq_order of PUQ} \end{equation}
	We then have 
	\begin{align*}
		&\sol{\sum_{u' \neq u}}\frac{|\pr{U|Q}{u|q} - \pr{U|Q}{u'|q}|}{2} \\
		&\quad\eqann{a} \sum_{u=0}^{L-1}u\pr{U|Q}{u|q} - \sum_{u=0}^{L-1}(L-1-u)\pr{U|Q}{u|q} \\ 
		&\quad\eqann{b} L - \sum_{u=0}^{L-1}(L-u)\pr{U|Q}{u|q} - \sum_{u=0}^{L-1}(L-1-u)\pr{U|Q}{u|q}\\
		&\quad= (L-1) - 2\sum_{u=0}^{L-1}(L-1-u) \pr{U|Q}{u|q}\\
		&\quad\leq (L-1) -2\sum_{u=0}^{L-1}\min\{1,L-1-u\} \pr{U|Q}{u|q}\\
		&\quad= (L-1) -2 \sum_{u=0}^{L-2} \pr{U|Q}{u|q}\\
		&\quad\eqann{c} (L-1) -2(1-\max_u\pr{U|Q}{u|q}).
	\end{align*}
	To see~\eqannref{a}, note that $|a-b| = \max\{a,b\}-\min\{a,b\}$. Using the ordering~\eqref{eq_order of PUQ}, we construct two $L\times L$ matrices: one with constant columns, with value $\pr{U|Q}{u|q}$ in column $(u+1)$,  and one with constant rows, with value $\pr{U|Q}{u|q}$ in row $(u+1)$, $u=0,1,\ldots,L-1$    . We compute the difference of the two matrices;  the desired sum equals the sum of elements above the diagonal. 
	Then, \eqannref{b} is because $\sum u \pr{U|Q}{u|q} + \sum (L-u) \pr{U|Q}{u|q} = L$, 
	and~\eqannref{c} is by the ordering~\eqref{eq_order of PUQ} and since $\sum_u \pr{U|Q}{u|q} = 1$. 
	Thus, for any $q \in \mathcal{Q}$,
	\begin{equation} 
	\sol{\sum_{u' \neq u}}\frac{|\pr{U|Q}{u|q} - \pr{U|Q}{u'|q}|}{2(L-1)} \leq 1 - \frac{2(1-\max_u\pr{U|Q}{u|q})}{L-1}	. \label{eq_TVL and PeL}
	\end{equation}
	Using~\eqref{eq_TVL and PeL} in the definition of $\TVL{\rv{U}|\rv{Q}}$ and recalling the expression for $\PeL{\rv{U}|\rv{Q}}$, we obtain the desired inequality. 
\end{IEEEproof}
The following corollary tightens~\cite[Proposition 3.2]{sasoglu_thesis}. 
\begin{corollary}
	The non-binary Bhattacharyya parameter upper-bounds the probability of error according to
	\[ \PeL{\rv{U}|\rv{Q}} \leq \frac{L-1}{2}\BPL{\rv{U}|\rv{Q}}.\]
\end{corollary}
\begin{IEEEproof}
	This is a consequence of the left-hand inequality of~\eqref{eq_TVL BPL ENTL bounds} and \Cref{lem_non binary TV and Pe}. 
\end{IEEEproof}

The non-binary distribution parameters are all natural extensions of their versions when $\rv{U}$ is binary. In particular, the non-binary parameters have the same form as their binary counterparts. As shown above, the consequences of \Cref{lem_TV distance bounds} apply to the non-binary parameters as well. 
They also satisfy \Cref{lem_effect of conditioning}; the extension of its proof is straightforward.  Thus, the non-binary distribution parameters may be used to define the relevant processes as in~\eqref{eq_defs of Kn Zn Hn} and~\eqref{eq_defs of Knhat Znhat Hnhat}. 

\subsection{Polarization of the Distribution Parameters}
In the binary case, fast polarization is obtained by \Cref{lem_simple proof}, which requires
polarization bounds on the Bhattacharyya and total variation distance processes. In the non-binary
case, the Bhattacharyya process and the total variation distance process are defined similarly to their binary
counterparts, with the relevant parameters replaced with their non-binary form presented above. 
The relevant polarization bounds for the non-binary Bhattacharyya process were obtained in~\cite[Lemma
3.5]{sasoglu_thesis}. We now establish polarization bounds for the total variation distance process
that extend~\Cref{prop_K is a supermartingale} to the non-binary case; we abuse notation and use
$\rv{K}_n$ to denote the non-binary counterpart of the total variation distance process.  \Cref{prop_K inequalities
for FAIM processes} is similarly extended; we omit the derivation.

\begin{proposition}
		Assume that $(\rv{X}_j,\rv{Y}_j)$, $j\in \mathbb{Z}$ is a memoryless process, where $\rv{X}_j \in \mathcal{U}$ such that $|\mathcal{U}| = L$, and $\rv{Y}_j \in \mathcal{Y}$. 
	Then, 
	\begin{equation}
		\rv{K}_{n+1} \leq \begin{dcases} \frac{2(L-1)}{L}\rv{K}_n^2 & \text{if } \rv{B}_{n+1} = 0 \\ 
 									  \left(1+\frac{L}{2}\right)\rv{K}_n  & \text{if } \rv{B}_{n+1} = 1. 
 									 \end{dcases}	\label{eq_polarization bounds for TV nonbinary}
	\end{equation} 
\end{proposition}
Observe that when $L=2$, the right-hand-side of~\eqref{eq_polarization bounds for TV nonbinary} coincides with that of~\eqref{eq_polarization bounds for TV}. 

\begin{IEEEproof}
	As in \Cref{prop_K is a supermartingale}, we fix $\rv{B}_1,\ldots, \rv{B}_{n}$ and let $i-1 = (\rv{B}_1 \rv{B}_{2} \cdots \rv{B}_n)_2$.  
This also fixes the value of $\rv{K}_n$. We denote $\prrv{\rv{U}_i,\rv{V}_i,\rv{Q}_i,\rv{R}_i}{u,v,q,r} = P(u,q)P(v,r)$. Slightly abusing notation, we further denote $P(u,q) = P(q)P(u|q)$. We set $\rv{T}_i = \rv{U}_i + \rv{V}_i$; this is modulo-$L$ addition, so \[\prrv{\rv{T}_i,\rv{V}_i,\rv{Q}_i,\rv{R}_i}{t,v,q,r} = P(t-v,q)P(v,r),\] where $t-v$ is computed modulo-$L$. 

We shall need the following inequality: 
\begin{equation}
\begin{split}
&\sol{\sum_{\underline{u'} \neq u}} |\pr{U|Q}{u|q} - \pr{U|Q}{u'|q}|\\
&\quad\eqann[\geq]{a} \left|\sol[r]{\sum_{\underline{u'} \neq u}} (\pr{U|Q}{u|q} - \pr{U|Q}{u'|q}) \right|	 \\ 
&\quad\eqann{b}  \left|(L-1)\pr{U|Q}{u|q} - (1-\pr{U|Q}{u|q}) \right| \\
&\quad= L \left|\pr{U|Q}{u|q} - \frac{1}{L}\right|.
\end{split}
\label{eq_sum uu' vs sum u}	
\end{equation}
Here, \eqannref{a} is by the triangle inequality and \eqannref{b} is because $\sum_u \pr{U|Q}{u|q} = 1$.

We compute $\rv{K}_{n+1}$ using~\eqref{eq_Kn single step polarization}. For the case $\rv{B}_{n+1} = 0$, note that
\begin{align*}
	&\sol{\sum_{t'\neq t}} \left|\sum_v P(v|r)(P(t-v|q)-P(t'-v|q)) \right| \\
	&\eqann{a} \sol{\sum_{t'\neq t}} \left|\sum_v \bigg(P(v|r)-\frac{1}{L}\bigg)\bigg(P(t-v|q)-P(t'-v|q)\bigg) \right| \\
	&\eqann[\leq]{b} \sol{\sum_{t'\neq t}} \sum_v \left|\bigg(P(v|r)-\frac{1}{L}\bigg)\bigg(P(t-v|q)-P(t'-v|q)\bigg) \right| \\
	&\eqann[=]{c} \sol{\sum_{t'\neq t}} \sum_v \bigg|P(v|r)-\frac{1}{L}\bigg|\cdot\bigg|P(t-v|q)-P(t'-v|q) \bigg|\\
	&=   \sum_v \bigg|P(v|r)-\frac{1}{L}\bigg|\cdot \sol{\sum_{t'\neq t}} \bigg|P(t-v|q)-P(t'-v|q) \bigg| \\ 
	&\eqann[\leq]{d}\frac{1}{L}  \sol{\sum_{v'\neq v}}\bigg|P(v|r)-P(v'|r) \bigg| \cdot \sol{\sum_{t'\neq t}} \bigg|P(t|q)-P(t'|q) \bigg|,
\end{align*}
where \eqannref{a} is because $\sum_v P(t-v|q) = \sum_v P(t'-v|q)$ for any $t,t'$, \eqannref{b} is by the triangle inequality, \eqannref{c} is because $|ab| = |a|\cdot|b|$, and \eqannref{d} is by~\eqref{eq_sum uu' vs sum u} and since the sum over $t,t'$ is unaffected by the shift in $v$. 
Thus,
\begin{align*}
&\TVL{\rv{T}_i|\rv{Q}_i,\rv{R}_i} \\ 
&= \sum_{q,r}\frac{P(q)P(r)}{2(L-1)} \sol{\sum_{t\neq t'}} \left|\sum_v P(v|r)(P(t-v|q)-P(t'-v|q)) \right| \\
&\leq \frac{2(L-1)}{L} \rv{K}_n^2. 
\end{align*}
Recalling~\eqref{eq_Kn single step polarization}, this proves the top inequality of~\eqref{eq_polarization bounds for TV nonbinary}. 

For the case $\rv{B}_{n+1} = 1$, note that by~\eqref{eq_abcd equality} and the triangle inequality, when $v'\neq v$ we have
\begin{align*}
&2\big|P(t-v|q)P(v|r)-P(t-v'|q)P(v'|r)\big| \\
&= \bigg|\Big(P(t-v|q)+P(t-v'|q)\Big)\Big(P(v|r)-P(v'|r)\Big) \\
&\quad+ \Big(P(v|r)+P(v'|r)\Big)\Big(P(t-v|q)-P(t-v'|q)\Big)\bigg| \\
&\leq \Big(P(t-v|q)+P(t-v'|q)\Big)\cdot \bigg|P(v|r)-P(v'|r)\bigg| \\
&\quad+ \Big(P(v|r)+P(v'|r)\Big)\cdot\bigg|P(t-v|q)-P(t-v'|q)\bigg| \\
&\leq \Big(P(t-v|q)+P(t-v'|q)\Big)\cdot \bigg|P(v|r)-P(v'|r)\bigg| \\
&\quad+ \bigg|P(t-v|q)-P(t-v'|q)\bigg|.
\end{align*}
The last inequality is due to the upper bound $P(v|r) + P(v'|r) \leq 1$ when $v' \neq v$. 
Hence, 
\begin{align*}
	&\sum_t \sol{\sum_{v'\neq v}} |P(t-v|q)P(v|r)-P(t-v'|q)P(v'|r)| \\
	&\leq \sol{\sum_{v'\neq v}} \big|P(v|r)-P(v'|r)\big| \\ &\quad  + 
	  \frac{1}{2}\sum_t  \sol{\sum_{v'\neq v}} \big|P(t-v|q)-P(t-v'|q)\big|\\
	&= \sol{\sum_{v'\neq v}} \big|P(v|r)-P(v'|r)\big| + \frac{L}{2} \sol{\sum_{t'\neq t}} \big|P(t|q)-P(t'|q)\big|.
\end{align*}
Consequently, 
\begin{align*}
&\TVL{\rv{V}_i|\rv{T}_i,\rv{Q}_i,\rv{R}_i} \\ 
&= \sol{\sum_{\substack{q,r,\\t}}} \sum_{v'\neq v} \frac{P(q)P(r)}{2(L-1)} |P(t-v|q)P(v|r)-P(t-v'|q)P(v'|r)|\\
&\leq \left(1 + \frac{L}{2}\right) \rv{K}_n.
\end{align*}
This proves the bottom inequality of~\eqref{eq_polarization bounds for TV nonbinary}.
\end{IEEEproof}

The bounds in~\eqref{eq_polarization bounds for TV nonbinary} are of the form required in \Cref{lem_simple proof}, allowing its use to establish fast polarization of the total variation distance process.

\section{Auxiliary Proofs for \Cref{sec_FAIM}} 
\label{app_proof of XY is psi mixing}
We denote $\rv{A}_j = (\rv{X}_j,\rv{Y}_j)$, $j \in\mathbb{Z}$, with realization $\alpha_j$, and $\rv{A}_M^N = (\rv{X}_M^N,\rv{Y}_M^N)$ with realization $\alpha_M^N$. For brevity, we denote  
$\prrv{\rv{A}_M^N}{} \equiv \prrv{\rv{A}_M^N}{\alpha_M^N}$, and similarly $\prrv{\rv{S}_N}{} \equiv \prrv{\rv{S}_N}{s_N}$. 

\begin{figure}
\begin{center}
\begin{tikzpicture}
	\node[draw, thick, minimum width = 3cm, minimum height = 0.9 cm] (box1) at (0,0) {$\rv{A}_1^L = (\rv{X}_1^L,\rv{Y}_1^L)$}; 
	\node[draw, thick, minimum width = 4cm, minimum height = 0.9 cm, right = 1 cm of box1] (box2) {$\rv{A}_{M+1}^N= (\rv{X}_{M+1}^N,\rv{Y}_{M+1}^N)$}; 
	\node[above = 0.1 of box1.north west] {$\rv{S}_0$};
	\node[below = 0.4 of box1.south west, anchor = base] {$a$};
	\node[above = 0.1 of box1.north east] {$\rv{S}_{L}$};
	\node[below = 0.4 of box1.south east, anchor = base] {$b$};
	\node[above = 0.1 of box2.north west] {$\rv{S}_M$};
	\node[below = 0.4 of box2.south west, anchor = base] {$c$};
	\node[above = 0.1 of box2.north east] {$\rv{S}_{N}$};
	\node[below = 0.4 of box2.south east, anchor = base] {$d$};
\end{tikzpicture}
\end{center}
\caption{Two blocks of a FAIM process, not necessarily of the same length. The initial state of the first block, $\rv{S}_0$, assumes value $a \in\mathcal{S}$. The final state of the first block, $\rv{S}_L$, assumes value $b \in\mathcal{S}$. The initial state of the second block, $\rv{S}_M$, assumes value $c \in\mathcal{S}$. The final state of the second block, $\rv{S}_N$, assumes value $d \in\mathcal{S}$.}\label{fig_two blocks LM}
\end{figure}
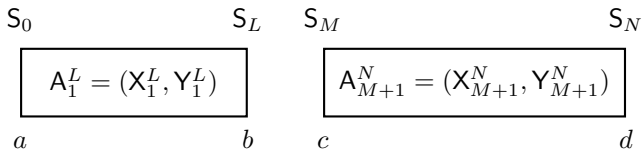

\begin{IEEEproof}[Proof of \Cref{lem_XY is psi-mixing}]
The function $\psi(N)$ was defined in~\eqref{eq_def of psi(n)}. We repeat the definition below using a notation that highlights the random variables at play. We deliberately do not use the notation~\eqref{eq_def of Q0 Qmn}, to explicitly show which random variables are being marginalized.
\[ \psi(N) = \begin{dcases} \max_{\istate,\mstate} \frac{\prrv{\rv{S}_N|\rv{S}_0}{\mstate|\istate}}{\prrv{\rv{S}_0}{\mstate}} & \text{if } N>0 \\[0.1cm]
 \max_{\istate} \frac{1}{\prrv{\rv{S}_0}{\istate}} &\text{if } N = 0. 	
 \end{dcases}
\] 
 Recall that by stationarity, $\prrv{\rv{S}_0}{} = \prrv{\rv{S}_N}{}$ for any $N$, so $\prrv{\rv{S}_N|\rv{S}_0}{\mstate|\istate}/\prrv{\rv{S}_0}{\mstate} = \prrv{\rv{S}_N|\rv{S}_0}{\mstate|\istate}/\prrv{\rv{S}_N}{\mstate}$. 

Since $\rv{S}_j$, $j=1,2,\ldots$ is an aperiodic and irreducible stationary finite-state Markov chain, $\psi(N)$ is non-increasing and $\psi(N) \to 1$ as $N\to \infty$. This is evident from the properties of such Markov chains; for a formal proof of this statement, see~\cite[Theorem 7.14]{bradley2007}. For such Markov chains $\prrv{\rv{S}_0}{\istate}>0$ for any $\istate \in \mathcal{S}$, so $\psi(0) < \infty$. 

It remains to show that $\prrv{\rv{A}_1^L,\rv{A}_{M+1}^N}{} \leq \psi(M-L) \prrv{\rv{A}_1^L}{}\prrv{\rv{A}_{M+1}^N}{}$. Consider first the case  $M>L$. Denote by $a,b,c,d$ the \emph{values} of states $\rv{S}_0,\rv{S}_L,\rv{S}_M,$ and $\rv{S}_N$, respectively (see \Cref{fig_two blocks LM}). Then,  
\begin{align*}
	&\prrv{\rv{A}_1^L,\rv{A}_{M+1}^N}{} \\
	&= \sol{\sum_{\alpha_{L+1}^M}} \prrv{\rv{A}_1^L,\rv{A}_{L+1}^M,\rv{A}_{M+1}^N}{} \\
				  &= \sol[l]{\sum_{\alpha_{L+1}^M}} \sol[r]{\sum_{d,a}} \prrv{\rv{A}_1^L,\rv{A}_{L+1}^M,\rv{A}_{M+1}^N,\rv{S}_N|\rv{S}_0}{} \prrv{\rv{S}_0}{}\\
				  &= \sol[l]{\sum_{\substack{d,c, \\ b, a}}} \sum_{\alpha_{L+1}^M} \prrv{\rv{A}_{M+1}^N,\rv{S}_N|\rv{S}_M}{}\prrv{\rv{A}_{L+1}^M,\rv{S}_M|\rv{S}_L}{} \prrv{\rv{A}_1^L,\rv{S}_L|\rv{S}_0}{} \prrv{\rv{S}_0}{}\\
				  &=\sol{\sum_{\substack{d,c, \\ b, a}}}  \prrv{\rv{A}_{M+1}^N,\rv{S}_N|\rv{S}_M}{}\left(\sum_{\alpha_{L+1}^M} \prrv{\rv{A}_{L+1}^M,\rv{S}_M|\rv{S}_L}{}\right) \prrv{\rv{A}_1^L,\rv{S}_L|\rv{S}_0}{} \prrv{\rv{S}_0}{}\\
				  &= \sol{\sum_{\substack{d,c, \\ b, a}}}  \prrv{\rv{A}_{M+1}^N,\rv{S}_N|\rv{S}_M}{} \prrv{\rv{S}_M|\rv{S}_L}{} \prrv{\rv{A}_1^L,\rv{S}_L|\rv{S}_0}{} \prrv{\rv{S}_0}{} \\
				  &= \sol{\sum_{\substack{d,c, \\ b, a}}}  \prrv{\rv{A}_{M+1}^N,\rv{S}_N|\rv{S}_M}{} \prrv{\rv{S}_M}{} \frac{\prrv{\rv{S}_M|\rv{S}_L}{}}{\prrv{\rv{S}_M}{}} \prrv{\rv{A}_1^L,\rv{S}_L|\rv{S}_0}{} \prrv{\rv{S}_0}{} \\
				  &\leq \psi(M-L) \sol{\sum_{\substack{d,c, \\ b, a}}}  \prrv{\rv{A}_{M+1}^N,\rv{S}_N|\rv{S}_M}{} \prrv{\rv{S}_M}{}  \prrv{\rv{A}_1^L,\rv{S}_L|\rv{S}_0}{} \prrv{\rv{S}_0}{} \\
				  &= \psi(M-L) \prrv{\rv{A}_1^L}{} \prrv{\rv{A}_{M+1}^N}{}. 
				  \end{align*}
				  
We proceed similarly for the case $M=L$. Again, $a$ and $d$ represent the \emph{values} of states $\rv{S}_0$ and $\rv{S}_N$. Both $b$ and $b'$ represent values of state $\rv{S}_L$; this distinction is to distinguish the summation variables of two different sums over values of $\rv{S}_L$. Thus,  
\begin{align*}
	\prrv{\rv{A}_1^L,\rv{A}_{L+1}^N}{} 
			   &= \sol{\sum_{\substack{a,b, \\ d}}} \prrv{\rv{A}_{L+1}^N,\rv{S}_N|\rv{S}_L}{} \frac{\prrv{\rv{S}_L}{}}{\prrv{\rv{S}_L}{}} \prrv{\rv{A}_1^L,\rv{S}_L|\rv{S}_0}{} \prrv{\rv{S}_0}{}\\
			   &\leq \psi(0) \sol{\sum_{d,b}} \prrv{\rv{A}_{L+1}^N,\rv{S}_N|\rv{S}_L}{} \prrv{\rv{S}_L}{} \cdot \left(\sum_{b',a}\prrv{\rv{A}_1^L,\rv{S}_L|\rv{S}_0}{} \prrv{\rv{S}_0}{}\right)\\
			   &= \psi(0) \prrv{\rv{A}_1^L}{} \prrv{\rv{A}_{L+1}^N}{};
				  \end{align*}
where the inequality is because $\prrv{\rv{A}_1^L,\rv{S}_L|\rv{S}_0}{} \leq \sum_{b'} \prrv{\rv{A}_1^L,\rv{S}_L|\rv{S}_0}{}$. 
\end{IEEEproof}

\begin{IEEEproof}[Proof of \Cref{lem_two adjacent blocks independent given shared state}]
Due to aperiodicity and irreducibility of the state sequence, $\prrv{\rv{S}_M}{\istate}>0$ for any $\istate \in \mathcal{S}$.
By the Markov Property, 
\begin{align*}
	\prrv{\rv{A}_1^M,\rv{A}_{M+1}^N|\rv{S}_M}{} &= \frac{\prrv{\rv{S}_M,\rv{A}_1^M,\rv{A}_{M+1}^N}{}}{\prrv{\rv{S}_M}{}}	\\
	&= \frac{\prrv{\rv{S}_M}{} \cdot \prrv{\rv{A}_1^M|\rv{S}_M}{} \cdot \prrv{\rv{A}_{M+1}^N|\rv{S}_M,\rv{A}_1^M}{}}{\prrv{\rv{S}_M}{}} \\
	&= \prrv{\rv{A}_1^M|\rv{S}_M}{} \cdot \prrv{\rv{A}_{M+1}^N|\rv{S}_M}{}.
\end{align*}
This proves~\eqref{eq_block independence given state N}. 

To derive~\eqref{eq_block independence given state 0 N M}, some more care is required to avoid division by $0$.  
By the Markov property, 
\begin{align*}
& \prrv{\rv{S}_0,\rv{S}_M,\rv{S}_N}{}\cdot \prrv{\rv{A}_1^M,\rv{A}_{M+1}^N|\rv{S}_0,\rv{S}_M,\rv{S}_N}{}\\
	&\quad=\prrv{\rv{S}_0,\rv{S}_M,\rv{S}_N,\rv{A}_1^M,\rv{A}_{M+1}^N}{}
\\ &\quad= \prrv{\rv{S}_0,\rv{S}_M}{} \cdot \prrv{\rv{A}_1^M|\rv{S}_0,\rv{S}_M}{} \cdot \prrv{\rv{S}_N,\rv{A}_{M+1}^N|\rv{S}_0,\rv{S}_M,\rv{A}_1^M}{} \\
	&\quad= \prrv{\rv{S}_0,\rv{S}_M}{} \cdot \prrv{\rv{A}_1^M|\rv{S}_0,\rv{S}_M}{} \cdot \prrv{\rv{S}_N,\rv{A}_{M+1}^N|\rv{S}_M}{} \\
	&\quad= \prrv{\rv{S}_0,\rv{S}_M}{} \cdot \prrv{\rv{A}_1^M|\rv{S}_0,\rv{S}_M}{} \cdot \prrv{\rv{S}_N|\rv{S}_M}{} \cdot \prrv{\rv{A}_{M+1}^N|\rv{S}_M,\rv{S}_N}{} \\
	&\quad= \prrv{\rv{S}_0,\rv{S}_M}{}\cdot \prrv{\rv{S}_N|\rv{S}_M,\rv{S}_0}{} \cdot \prrv{\rv{A}_1^M|\rv{S}_0,\rv{S}_M}{}  \cdot \prrv{\rv{A}_{M+1}^N|\rv{S}_M,\rv{S}_N}{} \\ 
	&\quad= \prrv{\rv{S}_0,\rv{S}_M,\rv{S}_N}{}  \cdot \prrv{\rv{A}_1^M|\rv{S}_0,\rv{S}_M}{}  \cdot \prrv{\rv{A}_{M+1}^N|\rv{S}_M,\rv{S}_N}{}.
\end{align*}
Recalling the definition of conditional probability~\cite[Section 33]{billingsley1995probability}, this implies~\eqref{eq_block independence given state 0 N M}. 
\end{IEEEproof}

\bibliographystyle{IEEEtran} 
\bibliography{mybib.bib} 
\end{document}